\documentclass[12pt,reqno]{amsart}
\usepackage[margin=1in]{geometry}

\pdfoutput=1

\usepackage{amsthm}
\usepackage{amssymb}
\usepackage{mathrsfs}
\usepackage{enumerate}
\usepackage{mathtools}
\usepackage[x11names]{xcolor}
\definecolor{darkred}{RGB}{139,0,0}
\definecolor{darkblue}{RGB}{0,0,139}
\definecolor{darkgreen}{RGB}{0,100,0}
\definecolor{darkmagenta}{RGB}{139,0,120}
\usepackage{enumitem}
\usepackage{tikz}
\usepackage{tikz-cd}
\usetikzlibrary{matrix,fit,calc,arrows,decorations.markings,3d,hobby,patterns}

\usepackage{booktabs}

\usepackage[linktocpage]{hyperref}
\hypersetup{
  colorlinks   = true,
  urlcolor     = darkred,
  linkcolor    = darkred,
  citecolor   = darkgreen}

\usepackage[noabbrev,capitalise]{cleveref}

\setenumerate[1]{label=(\roman*),}

\setitemize[1]{label=\raisebox{0.25ex}{\tiny$\bullet$}}

\providecommand\given{}
\newcommand\SetSymbol[1][]{%
  \nonscript\:#1\vert
  \allowbreak
  \nonscript\:
  \mathopen{}}
\DeclarePairedDelimiterX\Set[1]\{\}{%
  \renewcommand\given{\SetSymbol[]}
  #1
}

\newcommand{\nospaceperiod}{\makebox[0pt][l]{\,.}}

\newcommand{\bR}{\mathbb{R}}

\DeclarePairedDelimiter\ceil{\lceil}{\rceil}

\DeclarePairedDelimiter\abs{\lvert}{\rvert}

\usepackage{comment}
\usepackage{multirow}
\usepackage{caption}

\DeclareMathOperator{\Del}{\mathcal D}
\DeclareMathOperator{\SelDel}{\Del}
\DeclareMathOperator{\Cech}{\mathcal{C}}
\DeclareMathOperator{\DelCech}{\mathcal{D}\mathcal{C}}

\newcommand{\collapse}{\searrow}

\newcommand{\offset}{\mathcal{O}}
\newcommand{\cech}{\mathcal{C}}
\newcommand{\Inc}{\mathcal{I}}
\DeclareMathOperator{\rips}{\mathcal{R}}

\newcommand{\radius}[2]{\rho_{#1}}
\newcommand{\selradius}[2]{\rho{(#1, #2)}}
\newcommand{\meb}[1]{r _{#1}}

\newcommand{\deff}{\emph}

\newcommand{\R}{\mathbb{R}}

\newcommand\sbullet[1][.5]{\mathbin{\hbox{\scalebox{#1}{$\bullet$}}}}
\newcommand{\uni}{\sbullet[.5]}

\newcommand{\filt}{\uni}

\newtheorem{theorem}{Theorem}[section]
\newtheorem{lemma}[theorem]{Lemma}
\newtheorem{proposition}[theorem]{Proposition}
\newtheorem{corollary}[theorem]{Corollary}
\newtheorem*{corollary*}{Corollary}
\theoremstyle{definition}
\newtheorem{definition}[theorem]{Definition}

\theoremstyle{remark}
\newtheorem{remark}[theorem]{Remark}

\makeatletter
\renewcommand\subparagraph{\@startsection{subparagraph}{5}%
  \z@{.5\linespacing\@plus.7\linespacing}{-.5em}%
  {\normalfont\bfseries}}
\makeatother

\title[Delaunay Bifiltrations of Functions on Point Clouds]{Delaunay Bifiltrations of Functions on Point Clouds}
\author[Á.J. Alonso]{Ángel Javier Alonso}
\author[M. Kerber]{Michael Kerber}
\address{Institute of Geometry, Graz University of Technology, Austria}
\email{alonsohernandez@tugraz.at}
\email{kerber@tugraz.at}
\author[T. Lam]{Tung Lam}
\author[M. Lesnick]{Michael Lesnick}
\address{SUNY Albany, Albany, NY, USA}
\email{tlam@albany.edu}
\email{mlesnick@albany.edu}

\date{}

\begin{document}

\begin{abstract}
The Delaunay filtration $\Del_{\filt}(X)$ of a point cloud $X\subset \R^d$ is a
central tool of computational topology.  Its use is justified by the topological
equivalence of $\Del_{\filt}(X)$ and the offset (i.e., union-of-balls)
filtration of $X$.  Given a function $\gamma: X \to \R$, we introduce a Delaunay
bifiltration $\DelCech_{\filt}(\gamma)$ that satisfies an analogous topological
equivalence, ensuring that $\DelCech_{\filt}(\gamma)$ topologically encodes the
offset filtrations of all sublevel sets of $\gamma$, as well as the topological
relations between them.  $\DelCech_{\filt}(\gamma)$ is of size
$O(|X|^{\ceil{\frac{d+1}{2}}})$, which for $d$ odd matches the worst-case size
of $\Del_{\filt}(X)$.  Adapting the Bowyer-Watson algorithm for computing Delaunay triangulations, we give a simple, practical algorithm to compute $\DelCech_{\filt}(\gamma)$ in time $O(|X|^{\ceil{\frac{d}{2}}+1})$.  Our implementation, based on CGAL, computes $\DelCech_{\filt}(\gamma)$ with modest overhead compared to computing $\Del_{\filt}(X)$, and handles tens of thousands
of points in $\mathbb{R}^3$ within seconds.
\end{abstract}

\maketitle

\section{Introduction}
\subparagraph{Background and motivation.}
The \emph{offset filtration} of a finite set $X\subset\R^d$ is a multi-scale representation of $X$ which plays a major role in topological data analysis (TDA).
By definition, this is the nested family of spaces $\offset_{\filt}(X) = (\offset_r(X))_{r\geq 0}$, where %
$\offset_r(X)$ is the union of closed balls of radius $r$ centered at the points
in $X$.  The \emph{persistent homology} of $\offset_{\filt}(X)$ is one of the most
commonly computed objects in TDA.  However, $\offset_{\filt}(X)$ is not a
combinatorial object, and is therefore difficult to compute directly. So one
instead usually computes the \emph{Delaunay filtration}
$\Del_{\filt}(X) = (\Del_{r}(X))_{{r\geq 0}}$ (also called the
$\alpha$-filtration), a filtration of the Delaunay triangulation
$\Del(X)$ of $X$; see~\cref{fig:offset_delaunay}.   By a functorial variant of the classical nerve theorem~\cite{bauer2023unified}, $\Del_{\filt}(X)$ is topologically equivalent to $\offset_{\filt}(X)$ and therefore has the same persistent homology; see \cref{sec:bifiltrations} for a formal definition of topological equivalence.  For TDA computations with low-dimensional data, $\Del_{\filt}(X)$ is preferable to alternative simplicial filtrations such as the \v{C}ech and (Vietoris)-Rips filtrations because, as discussed below, it is asymptotically smaller and much more scalable in practice.
\begin{figure}
  \centering
  \includegraphics[page=1, scale=.50]{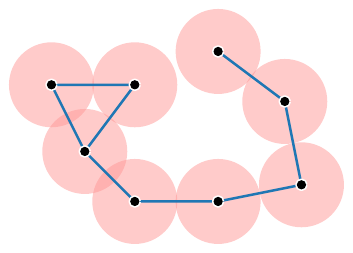}
  \qquad
  \includegraphics[page=2, scale=.50]{figures/offset.pdf}
  \caption{\small Offsets $\offset_{r}(X)\subset \offset_{r'}(X)$ (pink) and Delaunay complexes $\Del_{r}(X)\subset \Del_{r'}(X)$
 (blue) of a point set $X\subset \R^2$.}
  \label{fig:offset_delaunay}
\end{figure}

The index at which a simplex $\sigma$ first appears
 in $\Del_{\filt}(X)$ is defined to be the minimum radius of an empty circumsphere of $\sigma$.  If we instead define this index to be the radius of the minimum enclosing ball (meb) of $\sigma$, we obtain a topologically equivalent filtration, the \emph{Delaunay-\v{C}ech filtration} $\DelCech_{\filt}(X)$
 \cite{bauerMorseTheoryCech2016}.

The point set $X$ often comes equipped with a function $\gamma\colon X\to \R$.  One then wishes to refine the topological analysis of $X$ to take $\gamma$ into account.
In some cases, $\gamma$ is provided by an application; for example, if $X$ is the set of atom centers of a biomolecule, $\gamma$ may be a partial charge function \cite{cang2018representability,cang2020persistent,keller2018persistent}.
In other cases, $\gamma$ is constructed from $X$ itself, e.g., as an estimate of curvature or codensity (e.g., the inverse or negation of a density function).  The case where $\gamma$ is a codensity estimate is particularly important in TDA because in the presence of outliers, $\offset_{\filt}(X)$ can be insensitive to topological structure in the high-density regions of $X$~\cite{carlssonTheoryMultidimensionalPersistence2009,carlsson2008local}.%

Carlsson and Zomorodian~\cite{carlssonTheoryMultidimensionalPersistence2009} propose to analyze $\gamma$ using \emph{multiparameter persistence}, and specifically, to
consider the \emph{sublevel offset bifiltration}
$\offset_{\filt}(\gamma) = (\offset_{r,s}(\gamma))_{(r,s)\in A}$,
where $A=[0,\infty)\times\R$, $\offset_{r,s}(\gamma)=\offset_r(X^{s})$, and $X^{s}=\gamma^{-1}(-\infty,s]$ is the \emph{$s$-sublevel set} of $\gamma$.
This is indeed a bifiltration, in the sense that
$\offset_{r,s}(\gamma)\subset\offset_{r',s'}(\gamma)$ whenever $r\leq r'$ and
$s\leq s'$.  \cref{fig:offset} illustrates the sublevel offset bifiltration of a codensity function.%

As detailed below, multiparameter persistence is an active area of TDA, providing well-developed computational tools for working with bifiltrations and their persistent homology \cite{botnan2022introduction}.  One would like to apply these tools to $\offset_{\filt}(\gamma)$, but as with $\offset_{\filt}(X)$, direct computation of $\offset_{\filt}(\gamma)$ is difficult.
Thus, one wants an equivalent simplicial bifiltration that can be computed efficiently.  The goal of this paper is to introduce such a bifiltration.

\begin{figure}
  \centering
  \includegraphics[width=0.5\textwidth]{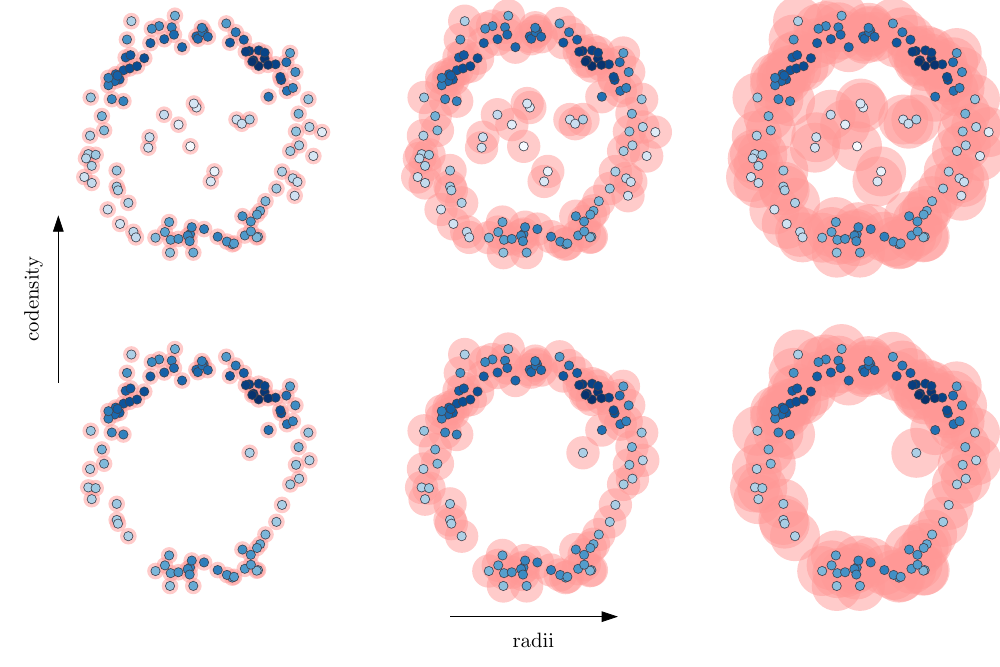}
  \caption{\small Part of the sublevel offset bifiltration $\offset_{\filt}(\gamma)$ of  a
     codensity function $\gamma\colon X\to\bR$, where $X\subset \R^2$.   Lightness of shading is proportional to the value of  $\gamma$.%
    }
  \label{fig:offset}
\end{figure}

\subparagraph{Contributions.}
In this work, we introduce the \emph{sublevel Delaunay-\v{C}ech bifiltration}
$\DelCech_{\filt}(\gamma) = (\DelCech_{r,s}(\gamma))_{(r,s)\in A}$, a bifiltered analogue of $\DelCech_{\filt}(X)$ which is topologically equivalent to $\offset_{\filt}(\gamma)$,
provably small in size, and readily computed in practice.
The definition of $\DelCech_{\filt}(\gamma)$ is not obvious: one might hope to define $\DelCech_{\filt}(\gamma)$ as a bifiltration of the Delaunay triangulation of $X$, but this fails because Delaunay triangulations
are not monotonic with respect to insertion of new points; see~\cref{fig:basic_delaunay}.
\begin{figure}
    \centering
    \includegraphics[page=1, scale=0.6]{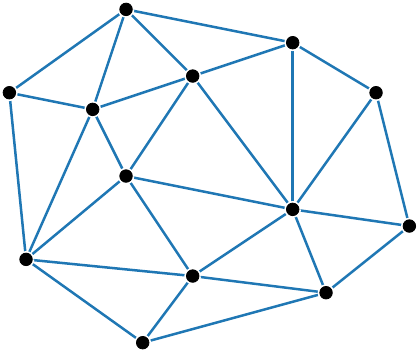}
    \hspace{5em}
    \includegraphics[page=2, scale=0.6]{figures/non_monotonicity.pdf}

    \caption{\small The Delaunay triangulations $\Del(X)$ and $\Del(Y)$ of
      planar point sets $X$ (left) and
      $Y = X\cup\Set{y}$ (right, with $y$ in red). Note that $\Del(X)\not\subset \Del(Y)$.}\label{fig:basic_delaunay}
\end{figure}
We instead define  $\DelCech_{\filt}(\gamma)$ as a bifiltration of a novel simplicial complex $\Inc(\gamma)$ with vertex set $X$, the  \emph{incremental Delaunay complex}; see \cref{def:incremental_delaunay}.  $\Inc(\gamma)$ contains the Delaunay triangulation of every sublevel set of $\gamma$, as well as additional simplices needed to encode the topological relationships between consecutive levels.  $\mathcal I(\gamma)$ embeds in $\R^{d+1}$ via the celebrated parabolic lift paradigm for Delaunay triangulations.

We characterize $\Inc(\gamma)$ in terms of two standard (and closely related) incremental constructions of Delaunay triangulations: the flipping algorithm~\cite{edelsbrunner1996incremental,lawsonPropertiesNdimensionalTriangulations1986}
and Bowyer-Watson algorithm~\cite{bowyerComputingDirichletTessellations1981,watsonComputingNdimensionalDelaunay1981}.
These characterizations make clear that both algorithms implicitly compute $\Inc(\gamma)$.  In fact, an implementation of either algorithm can be used off the shelf to explicitly compute $\Inc(\gamma)$ with just a few additional lines of code.  Our code\footnote{\url{https://bitbucket.org/mkerber/function_delaunay}} computes $\Inc(\gamma)$ using \textsc{CGAL}'s implementation of the Bowyer-Watson algorithm in arbitrary dimensions~\cite{cgal_dt_dd} and \textsc{Gudhi}'s simplex tree data structure for simplicial complexes~\cite{simplex_tree_paper,gudhi:FilteredComplexes}.

Leveraging the characterization of $\Inc(\gamma)$ in terms of flips, we observe
that $\Inc(\gamma)$ is of size $O(|X|^{\ceil{\frac{d+1}{2}}})$, assuming $X$ in
general position. The Delaunay triangulation $\Del(X)$ has size
$\Theta(|X|^{\ceil{\frac{d}{2}}})$ in the worst case, so for $d$ odd, the worst-case
sizes of $\Del(X)$ and $\Inc(\gamma)$ are asymptotically equal.  We use our size bound for $\Inc(\gamma)$ to show that computation of $\Inc(\gamma)$ requires time $O(|X|^{\ceil{\frac{d}{2}}+1})$.

Given $\Inc(\gamma)$, we define $\DelCech_{\filt}(\gamma)$ by taking $\sigma\in
\Inc(\gamma)$ to appear in $\DelCech_{\filt}(\gamma)$ at index $(r_\sigma,s_\sigma)\in A$, where $r_\sigma$ is the radius of the meb of $\sigma$ and $s_\sigma=\max_{x\in \sigma} \gamma(x)$.  Alternatively, we can take each $\sigma\in \Inc(\gamma)$ to appear at index $(\rho_\sigma,s_\sigma)$, where $\rho_\sigma$ is the radius
of a smallest sphere witnessing that $\sigma\in\Inc(\gamma)$; see \cref{def:incremental_delaunay}.  We call the resulting bifiltration $\Del_{\filt}(\gamma)$ the \emph{sublevel Delaunay bifiltration}; this is a bifiltered analogue of $\Del_{\filt}(X)$. Both choices of radii can be efficiently computed, but since the computation of meb radii is conceptually simpler and efficient software is readily available for this\footnote{\url{https://people.inf.ethz.ch/gaertner/subdir/software/miniball.html}}~\cite{cgal_meb, gaertner-esa}, our experiments focus on $\DelCech_{\filt}(\gamma)$.  We prove that
$\Del_{\filt}(\gamma)$, $\DelCech_{\filt}(\gamma)$ and $\offset_{\filt}(\gamma)$ are topologically
equivalent using Bauer and Edelsbrunner's \textit{selective Delaunay complexes}~\cite{bauerMorseTheoryCech2016} and discrete Morse theory.

Our computational experiments demonstrate that our algorithm for computing $\DelCech_{\filt}(\gamma)$ is highly effective in practice and, for low-dimensional Euclidean data, is far more scalable than the \emph{sublevel Rips bifiltration}, a standard alternative construction discussed below.  For instance, for $32000$ points in $\R^3$,
$\DelCech_{\filt}(\gamma)$ is computed in at most half a minute.  This enables the application of 2-parameter persistence to much larger low-dimensional data sets than previously possible. %

\subparagraph{Related work.}
Besides $\Del_{\filt}(X)$, two other simplicial filtrations are commonly constructed from a point cloud $X\subset \R^d$, the \emph{\v{C}ech} and \emph{Rips filtrations} $\cech_{\filt}(X)$ and $\rips_{\filt}(X)$. %
Unlike $\Del_{\filt}(X)$,  both the \v{C}ech and Rips filtrations extend straightforwardly to \emph{sublevel \v Cech} and \emph{sublevel Rips bifiltrations}  $\cech_{\filt}(\gamma)$, $\rips_{\filt}(\gamma)$ of a function $\gamma\colon X\to \R$, in essentially the same way that $\offset_{\filt}(X)$ extends to $\offset_{\filt}(\gamma)$  \cite{carlssonTheoryMultidimensionalPersistence2009}.
Like $\Del_{\filt}(X)$, $\cech_{\filt}(X)$ is topologically equivalent to $\offset_{\filt}(X)$, and by Jung's theorem \cite{jung1901ueber}, $\rips_{\filt}(X)$ is a 2-approximation to $\offset_{\filt}(X)$ (in the log-homotopy interleaving distance, see \cite{blumberg2017universality}); both results extend to the corresponding bifiltrations.  $\rips_{\filt}(X)$ and $\rips_{\filt}(\gamma)$ have the advantage that their definition and computation extend naturally to arbitrary finite metric spaces.

The $k$-skeleta of both $\rips_{\filt}(X)$ and $\cech_{\filt}(X)$ %
 have size $\Theta(|X|^{k+1})$, and the same is true for $\cech_{\filt}(\gamma)$ and $\rips_{\filt}(\gamma)$.  Therefore, one computes the $k$-skeleta of these objects only for very small $k$ (say, $k\leq 3$), which is sufficient for homology computations in dimension up to $k-1$.
 Even so, the size of these (bi)filtrations can impede practical use.  Therefore, for small values of $d$, %
it is usually preferable to work with $\Del_{\filt}(X)$ rather than $\rips_{\filt}(X)$ or $\cech_{\filt}(X)$.

 That said, highly optimized methods have been developed for computing persistent homology of the (1-parameter) Rips filtrations of arbitrary finite metric spaces \cite{bauerRipserEfficientComputation2021,glisse2022swap,boissonnat2020edge}, and these are used widely.  Recent work has focused on adapting some such optimizations to the computation of sublevel Rips bifiltrations \cite{bauer2023efficient,alonso2023filtration}.   Sublevel Rips bifiltrations have recently been applied to problems in computational chemistry \cite{keller2018persistent} and cancer imaging \cite{vipond2021multiparameter,benjamin2022multiscale}.  As these applications involve low-dimensional Euclidean data, future applications along similar lines may benefit from instead using the Delaunay bifiltrations introduced in this paper.
 The standard approach to computing $\cech_{\filt}(X)$ \cite{gudhi:CechComplex} adapts to the computation of $\cech_{\filt}(\gamma)$, though to our knowledge this has not yet been implemented.

The problem of non-monotonicity of Delaunay triangulations has been considered in prior work by Reani and Bobrowski~\cite{reaniCoupledAlphaComplex2021}.
Given two point sets $X,Y\subset \R^d$, they introduce a computable \emph{coupled Delaunay filtration} $\mathcal A_{\filt}(X,Y)$ containing both $\Del_{\filt}(X)$ and $\Del_{\filt}(Y)$ as subfiltrations
 such that the zigzag \[\Del_{\filt}(X)\hookrightarrow \mathcal A_{\filt}(X,Y)\hookleftarrow \Del_{\filt}(Y)\] is topologically equivalent to
\[\offset_{\filt}(X)\hookrightarrow \offset_{\filt}(X\cup Y)\hookleftarrow \offset_{\filt}(Y).\]
See also~\cite{dimontesano2022persistent,biswas2022size} for a recent extension to three or more point sets.   In the special case that $Y=X\cup\{p\}$, our main construction yields a filtration much smaller than  $\mathcal A_{\filt}(X,Y)$ with the same property.  More generally, a construction similar to ours yields a smaller filtration whenever $X\cap Y\ne \varnothing$; we leave further discussion of this (including computational considerations) to future work.

When $\gamma$ is a codensity function, the sublevel bifiltrations of $\gamma$ enable data analysis that is sensitive to both density and spatial scale.  However, one disadvantage of working with these bifiltrations is that such $\gamma$ always depends on a choice of \emph{bandwidth parameter}, and it is often unclear how to select this.  Several density-sensitive bifiltrations have been proposed that do not depend on a bandwidth parameter \cite{lesnick2015interactive,blumberg2022stability,sheehy2012multicover}, but these are more computationally costly than either $\rips_{\filt}(\gamma)$ or $\Del_{\filt}(\gamma)$.  Most closely related to our work is the \emph{rhomboid bifiltration}, introduced in \cite{edelsbrunner2021multi} and studied in \cite{corbet2023computing}.  This is a polyhedral extension of $\Del(X)$ which satisfies a strong robustness property \cite{blumberg2022stability} and is computable in polynomial time \cite{edelsbrunner2023simple}.  However, the rhomboid bifiltration of $X\subset \R^d$ has exactly ${|X|}\choose{d+1}$ top-dimensional cells, which means it is usually too large to fully compute.

In the computational pipeline of multiparameter persistence, our task of obtaining a bifiltration from a data set is just the first step; one then usually aims to compute and visualize \emph{topological invariants} of the bifiltration~\cite{lesnick2015interactive,carlssonTheoryMultidimensionalPersistence2009,cai2021elder,scoccola2023persistable,dey_et_al:LIPIcs.SoCG.2022.34}.  Among the many invariants that have been proposed, \emph{signed barcodes} are among the most actively studied~\cite{botnan2021signed,kim2021generalized,botnan2022bottleneck,mccleary2022edit,asashiba2019approximation,blanchette2021homological,morozov2021output}. There is also substantial interest in invariants taking values in linear spaces, since these are readily incorporated into standard machine learning and statistics pipelines \cite{loiseaux2023stable,vipond2018multiparameter,corbet2019kernel,carriere2020multiparameter,xin2023gril}.
One can also work with 2-parameter persistent homology via \emph{distance functions} \cite{kerber2020efficient,adcock2014classification,bjerkevik2021ell,keller2018persistent}. All of these approaches are computationally challenging and are most practical when the input bifiltration is small; this is the practical motivation of our work.
As an intermediate step in the computation of invariants, it is often beneficial to employ algebraic compression techniques to replace the bifiltration with a smaller object,
for instance by computing a minimal representation of homology. Such compressed representations can often be computed with surprising efficiency in the 2-parameter setting \cite{fkr-compression,lw-computing,alonso2023filtration,bauer2023efficient}. Our approach can be easily combined with these techniques,
as we show in our experimental evaluation.

\subparagraph{Outline.}
In Section~\ref{sec:incr_del_cmplx}, we define the incremental Delaunay complex, discuss its computation, and bound its size.
In Section~\ref{sec:bifiltrations}, we formally define topological equivalence and state the main topological equivalence results of this work, which are proven in~\cref{sec:proof_of_collapsing_theorem}.
In Section~\ref{sec:implementation}, we report on our experimental
evaluations.
We conclude in Section~\ref{sec:conclusion}.

\section{The incremental Delaunay complex}\label{sec:incr_del_cmplx}

A circumsphere of $Q\subset \bR^{d}$ is a $(d-1)$-sphere $S$ such that all points in $Q$ lie on $S$.
Throughout the paper, we assume that our point sets are in general position,
meaning that every non-empty subset $Q\subset X$ of at most $d+1$ points is affinely independent, and no point of $X\setminus Q$ lies
on the smallest circumsphere of $Q$.
The \emph{Delaunay triangulation} of $X$ is the simplicial complex $\Del(X)$
whose simplices are non-empty subsets $\sigma\subset X$ having an empty circumsphere $S$, that is, no points of $X$ lies inside $S$.

Given a function $\gamma\colon X\to \R$, we totally order $X$ by values of $\gamma$, breaking any ties arbitrarily.
We write $X=\{x_1,x_2,\ldots,x_{|X|}\}$ and for $i\in \{1,\ldots,|X|\}$, we let $X_{i}\subset X$ denote the first $i$ points of $X$.  For a simplex $\sigma\subset X$, let $\max(\sigma)$ be the maximum
point in $\sigma$.

Our definition of the incremental Delaunay complex $\Inc(\gamma)$ will depend only on the order on $X$, and not on the specific values taken by $\gamma$.  Hence, we write $\Inc(X)=\Inc(\gamma)$.  %

\begin{definition}
\label{def:incremental_delaunay}
 The \deff{incremental Delaunay complex}  $\Inc(X)$
  is the simplicial complex whose simplices are the non-empty subsets $\sigma\subset X$ such that there exists a circumsphere $S$ of $\sigma\setminus\max(\sigma)$ satisfying the following two properties:
  \begin{itemize}
  \item  $\max(\sigma)$ is either inside or on $S$ and,
  \item  each $x_i<\max(\sigma)$ is either outside or on $S$.
  \end{itemize}
  We call the infimal radius of such $S$ the \deff{incremental Delaunay radius} of $\sigma$, and denote it as $\rho_\sigma$.
\end{definition}
Thus, writing $x_k:=\max(\sigma)$, the sphere $S$ of
\cref{def:incremental_delaunay} must be empty of points in $X_{k-1}$ and $x_k$
must be either on or inside $S$.  There exists such $S$ with $x_k$ on $S$ if and
only if $\sigma\in \Del(X_k)$.  It follows that $\bigcup_{i=1}^{\abs{X}} \Del(X_{i})\subset \Inc(X)$.
Note also that $\Inc(X_i)\subset\Inc(X_j)$ whenever $i\leq j$, i.e., the incremental Delaunay complex is monotone under point insertions.

\subparagraph{Computation.}
We now explain how $\Inc(X)$ can be computed via the Bowyer-Watson algorithm.
Besides containing $\Del(X_{i})$ for each $i$, $\Inc(X)$ usually contains additional simplices of dimension $d+1$ and their faces.
To elaborate, for a $d$-simplex $\sigma\in\Del(X_i)$, we will say that the pair $(\sigma,x_{i+1})$ is a \emph{conflict pair} if $x_{i+1}$ is inside the circumsphere of $\sigma$.  We also call $\sigma$ an \emph{$i$-conflict}, or simply a \emph{conflict}.
By definition, we then have $\sigma\cup\{x_{i+1}\}\in\Inc(X)$, and all $(d+1)$-simplices of $\Inc(X)$ arise from conflict pairs in this way. We summarize this finding:
\begin{lemma}\label{prop:bw}
The $(d+1)$-simplices in $\Inc(X)$ are in one-to-one-correspondence with conflict pairs.
\end{lemma}
We would like to leverage \cref{prop:bw} to compute $\Inc(X)$, by computing all conflict pairs.  But in general, the $(d+1)$-simplices of $\Inc(X)$ and their faces do not yield the entire complex $\Inc(X)$.
The reason is that a point $x_{i+1}$
might lie outside the convex hull of $X_i$, without giving rise to any conflict pair.
Such cases can be avoided if we assume that the simplex spanned by the first $d+1$ points of $X$ is the convex hull of $X$;
we call this the \emph{$\Delta$-property}.  As we explain in \cref{rem:Delta_Prop} below, for the purpose of computing $\Inc(X)$, this assumption entails no loss of generality.

Assuming the $\Delta$-property, any $x_{i+1}$ with $i+1>d+1$
lies in the convex hull of the previous points, and thus inside some $d$-simplex $\sigma$ of $\Del(X_i)$,
which implies that $(\sigma,x_{i+1})$ is a conflict pair.  It is then plausible that
the $(d+1)$-simplices of $\Inc(X)$ fully determine $\Inc(X)$.  The next lemma asserts that this is indeed true:
\begin{lemma}\label{lem:purity_lemma}
If $|X|\geq d+2$ and $X$ has the $\Delta$-property, then every simplex of $\Inc(X)$
is a face of a $(d+1)$-simplex.
\end{lemma}

\begin{proof}

Consider the simplex $\alpha=[x_1,\ldots,x_{d+1}]$.  %
By the $\Delta$-property and the fact that $|X|\geq d+2$, the point $x_{d+2}$ lies in $\alpha$.  Hence $(\alpha,x_{d+2})$ is a conflict pair, so $\alpha$ is a  face of the simplex $[x_1,\ldots,x_{d+2}]$ in $\Inc(X)$.
Now assume for a contradiction that a simplex $\tau\in\Inc(X)$  with $\dim(\tau)<d+1$ is maximal (i.e., has no proper coface).  Let $x_{i+1}=\max(\tau)$ and $\tau'=\tau\setminus \{x_{i+1}\}$. As $\alpha$ is not maximal, we have $i+1\geq d+2$.

We claim that $\tau\in\Del(X_{i+1})$.  For all upcoming notation, see Figure~\ref{fig:pencil} for illustrations.  By the definition of $\Inc(X)$, there is a circumsphere $S$
of $\tau'$ such that  $x_{i+1}$ is either on or inside $S$, and no point of $X_i$ lies inside $S$.  %
Let $S'$ be a circumsphere of $\tau$; this exists by general position and the assumption that $\dim(\tau)<d+1$.  We will show that no point of $X_i$ lies inside $S'$, which proves the claim.
Since both $S$ and $S'$ are circumspheres
of $\tau'$, there exists a collection $(S_t)_{t\in [0,1]}$ of circumspheres of $\tau'$ interpolating between $S$ and $S'$; that is, $S_0=S$, $S_1=S'$, and the radius and center of $S_t$ both vary continuously with $t$.  We call this collection a \emph{pencil}.  In fact, we can choose the pencil so that $x_{i+1}$ lies on or inside each $S_t$.
Recall that no point of $X_i$ lies inside $S$.  Thus, if a point of $X_i$ lies inside $S'$, continuity of the pencil implies that there is a circumsphere $S''$ in the pencil for which a point $p$ of $X_i\setminus\tau$ lies on $S''$ and no point of $X_i$ lies inside $S''$.  Then, since $x_{i+1}$ lies on or inside $S''$, $\tau\cup \{p\}$ is a simplex of $\Inc(X)$, and hence $\tau$ is not maximal, a contradiction.  We conclude that no point of $X_i$ lies inside $S'$, and hence that $\tau\in\Del(X_{i+1})$, as claimed.

\begin{figure}
\centering
\includegraphics[width=6cm]{figures/pencil.eps}
\caption{\small Illustration of the proof of Lemma~\ref{lem:purity_lemma}. Here, $X_i=\{a,b,p\}$ and $\tau=\{a,b,x_{i+1}\}$. Note that $S$ contains
no point of $X_i$ by assumption, whereas $S'$ might contain such points (like $p$ in the figure).}
\label{fig:pencil}
\end{figure}

Since $\tau\in\Del(X_{i+1})$ and all maximal simplices of $\Del(X_{i+1})$ are of dimension $d$, we have $\dim(\tau)=d$.
The sphere $S'$ also witnesses that $\tau'\in\Del(X_i)$. Since $\tau'$ is a $(d-1)$-simplex, there exist either two $d$-dimensional cofaces
$\sigma, \sigma'$ of $\tau$ in $\Del(X_i)$, or (when $\tau'$ is a convex hull simplex) only one such coface $\sigma$.
In the former case, since $\tau$ is coface of $\tau'$ in $\Del(X_{i+1})$, one
of the $d$-simplices, say $\sigma$, is not in $\Del(X_{i+1})$.
In the latter case, $\sigma$ is not in $\Del(X_{i+1})$ either, because otherwise $\tau'$ would not be on the convex hull
anymore, contradicting the $\Delta$-property.
In either case, this implies that $x_{i+1}$ lies inside the circumsphere
of $\sigma$, and hence $(\sigma,x_{i+1})$ is a conflict pair. However, since $\tau'$ is a face of $\sigma$,
the simplex $\sigma\cup\{x_{i+1}\}$ is a coface of $\tau$, a contradiction.
\end{proof}

\cref{prop:bw,lem:purity_lemma} imply that if $X$ has the $\Delta$-property, then to compute $\Inc(X)$ it suffices to determine all conflict pairs.  These pairs are computed by the Bowyer-Watson algorithm~\cite{bowyerComputingDirichletTessellations1981,watsonComputingNdimensionalDelaunay1981}, a standard incremental algorithm for Delaunay triangulation computation, which we now recall in outline; see also ~\cref{fig:bowyer_watson}.  Given $\Del(X_i)$, to compute $\Del(X_{i+1})$ the algorithm computes all $i$-conflicts. These simplices are removed from the triangulation, leaving an untriangulated star-convex region $R$
with center $x_{i+1}$.  $R$ is then retriangulated by the simplicial cone with base the boundary of $R$ and apex $x_{i+1}$, yielding $\Del(X_{i+1})$.

To compute all $i$-conflicts, we must first compute one of them; we call this step \emph{conflict location}.  Since $R$ is connected, once we have found a single $i$-conflict, we can efficiently find the remaining ones by searching the adjacency graph of the $d$-simplices of $\Del(X_i)$.  A naive conflict location strategy (which nevertheless suffices for our complexity analysis below) is to iterate through the $d$-simplices
of $\Del(X_i)$ until a conflict is found.  Alternatively, a standard approach is to maintain a
data structure for efficient point location, which can then be queried for the $d$-simplex containing $x_{i+1}$.  Several such data structures have been proposed and implemented~\cite{devillers-delaunay,dpt-walking,bt-randomized,boissonat2009incremental}.  Our code uses CGAL's implementation of point location.

\begin{remark}\label{rem:Delta_Prop}
To ensure that the $\Delta$-property holds, we can augment $X$ with a set $T$ of $(d+1)$
extra points ordered below $X$.  Choosing the points of $T$ to be sufficiently far away from $X$ ensures that
$\Inc(X)$ is the subcomplex of $\Inc(X\cup T)$ obtained by removing all points of $T$ and their cofaces.  This is a standard technique in geometric algorithms; see for instance~\cite[Sec 9.3]{dutch_book}.  $T$ need not be chosen explicitly, but can be handled symbolically; in fact, it suffices to add a single \emph{point at infinity}, as in  \cite{boissonnat2000triangulations}.
\end{remark}

\begin{figure}[h]
  \centering
 \includegraphics[page=1, scale=0.6]{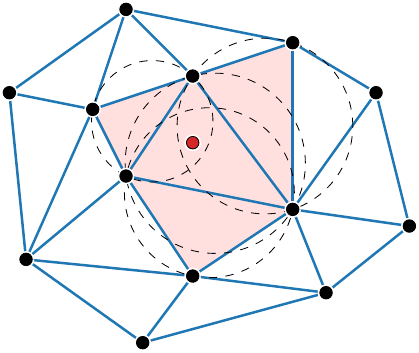}
 \hspace{1.5em}
  \includegraphics[page=2, scale=0.6]{figures/bowyer_watson.pdf}
  \hspace{1.5em}
  \includegraphics[page=3, scale=0.6]{figures/bowyer_watson.pdf}
  \caption{\small The Bowyer-Watson algorithm.  Left: $\Del(X_i)$ (blue) and
    $x_{i+1}$ (red). The triangles in conflict with $x_{i+1}$ are shaded.  Center: These triangles are removed.  Right: the hole is retriangulated, yielding $\Del(X_{i+1})$.}\label{fig:bowyer_watson}
\end{figure}

\subparagraph{Size of $\Inc(X)$.}
Assuming the $\Delta$-property and that $d$ is constant, \cref{lem:purity_lemma,prop:bw} together imply that $\Inc(X)$ has $O(m)$ simplices, where $m$ is the total number of conflict pairs encountered by Bowyer-Watson algorithm.
As we explain below,
\begin{proposition}
\label{eq:conflict_pairs_bound}
$m=O(|X|^{\lceil(d+1)/2\rceil})$  in the worst case.
\end{proposition}
\cref{eq:conflict_pairs_bound} implies the following, which in view of \cref{rem:Delta_Prop}, holds even when $X$ does not satisfy the $\Delta$-property:

\begin{theorem}\label{Thm:Size_Bound}
$\Inc(X)$ has $O(|X|^{\lceil(d+1)/2\rceil})$ simplices.
\end{theorem}

The worst-case size of $\Del(X)$ is $O(|X|^{\lceil d/2\rceil})$, which for odd dimensions matches our size bound for $\Inc(X)$.

As we will explain, \cref{eq:conflict_pairs_bound} follows from a result of Edelsbrunner and Shah~\cite{edelsbrunner1996incremental} about \emph{incremental flipping}, another well-known algorithm for computing Delaunay triangulations.  Alternatively, we can easily  deduce \cref{eq:conflict_pairs_bound} by adapting the proof of \cite[Theorem 3.10]{boissonnat2018geometric}, a complexity bound for an incremental convex hull algorithm, but we will not discuss this approach.

Before proving \cref{eq:conflict_pairs_bound}, we recall the incremental flipping algorithm: Assuming the $\Delta$-property for $X$ again, the algorithm computes $\Del(X_{i+1})$ from $\Del(X_i)$ as follows (see \cref{fig:flips}): The convex hull of $d+2$ points in general position has exactly two triangulations; the replacement of one with another is called a \emph{flip}.  We first subdivide the $d$-simplex of $\Del(X_i)$ containing $x_{i+1}$
into $d$+1 $d$-simplices, each with $x_{i+1}$ as a vertex.  We then transform
the resulting triangulation into $\Del(X_{i+1})$ via flips in the following way.  Let $T$ be the triangulation we maintain during the algorithm.  For a
$d$-simplex $\sigma$ in $T$ with vertex $x_{i+1}$, let $\hat \sigma=\sigma\setminus \{x_{i+1}\}$.  There is at most one vertex $y_{\sigma} \neq x_{i+1}$ such that $\hat \sigma\cup \{y_{\sigma} \}\in T$.  We say $\sigma$ is \emph{flippable} if $y_{\sigma} $ 
is inside the circumsphere of $\sigma$.  While there exists a flippable simplex $\sigma$ in $T$, we update $T$ by doing a flip on the restriction
of $T$ to $\sigma\cup\{x_{i+1},y_{\sigma} \}$.  When there are no more flippable 
simplices, we have $T=\Del(X_{i+1})$ and the computation is complete.  %

\begin{figure}[h]
  \centering
 \includegraphics[page=1, scale=0.6]{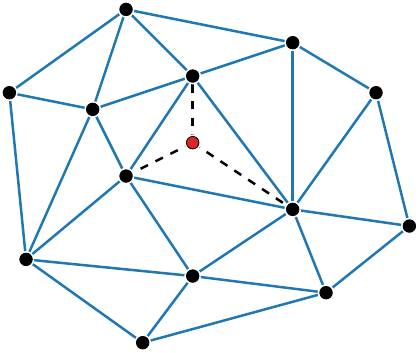}
 \hspace{1.5em}
  \includegraphics[page=2, scale=0.6]{figures/flips.pdf}
    \hspace{1.5em}
  \includegraphics[page=3, scale=0.6]{figures/flips.pdf}
  \caption{\small Incremental flipping. Left: we insert a new point $x_{i+1}$ (red) into $\Del(X_i)$ (blue); the
    dashed edges are the result of the initial subdivision step.  Center: the highlighted
    triangle is in conflict with $x_{i+1}$ and needs to be flipped; the dashed edge is inserted and the solid edge it crosses is removed. Right: we continue flipping the rest of triangles that are
    in conflict.}\label{fig:flips}
\end{figure}

\begin{proof}[Proof of \cref{eq:conflict_pairs_bound}]
It is easily checked that for each flippable simplex $\sigma$ encountered by the incremental flipping algorithm, we have $\hat\sigma \cup \{y_{\sigma}\}\in \Del(X_i)$ and therefore, $(\hat\sigma \cup \{y_{\sigma}\},x_{i+1})$ is a conflict pair.  Conversely, if a simplex $(\tau,x_{i+1})$ is a conflict pair, then $\tau\cup\{x_{i+1}\}$ must be involved in a flip.  Thus, as no flip is done more than once~\cite[Sec 8]{edelsbrunner1996incremental}, the total number of flips done by the algorithm equals the total number of conflict pairs. %
The number of flips is shown in~\cite[Sec 9]{edelsbrunner1996incremental} to be $O(|X|^{\lceil(d+1)/2\rceil})$ in the worst case, by counting simplices in dimension $\lfloor\frac{d}{2}\rfloor$.
\end{proof}

\subparagraph{Cost of computing $\Inc(X)$.}
Computing $\Inc(X)$ via Bowyer-Watson requires time $O(|X|^{\lceil d/2\rceil+1})$.  To explain, the cost is dominated by the cost of finding all conflicts.  Once we have located an $i$-conflict, the cost of finding the remaining $i$-conflicts is linear in the number of $i$-conflicts.  Thus, since $\Inc(X)$ has one $(d+1)$-simplex for each conflict, the cost of finding all conflicts admits an output-sensitive bound of $O(m+n)$, where $m=O(|X|^{\lceil (d+1)/2\rceil})$ is the size of $\Inc(X)$ and $n$ is the total cost of locating an $i$-conflict for each $i$.  Using the naive strategy of iterating through the simplices of $\Del(X_i)$, locating an $i$-conflict requires time $O(n_i)=O(|X_i|^{\lceil d/2\rceil})$, where $n_i$ is the number of $d$-simplices in $\Del(X_i)$.  Hence $n=O(\sum_i n_i)=O(|X|^{\lceil d/2\rceil+1})$  and $m+n$ satisfies the same bound.

We are not aware of any conflict location strategy in the literature that
achieves a better bound on $n$.  The key difficulty is that to compute
$\Inc(X)$, the order of point insertion must be compatible with $\gamma$.  In
contrast, we can compute $\Del(X)$ by inserting points in arbitrary order.  A
standard analysis shows when the points are inserted in random order,
Bowyer-Watson computes $\Del(X)$ in expected time $O(|X|^{\lceil d/2\rceil}+|X|\log |X|)$; the randomness lowers both the cost of point location and the number of conflicts, in expectation.  But we cannot use a random insertion order to compute $\Inc(X)$.

\subparagraph{The lifting perspective.}
Define $f\colon \R^d\to \R^{d+1}$ by $f(x)=(x, \|x\|^2)$.  It is well known that $\Del(X)$ is the image of the lower hull of $f(X)$ under the projection $\R^{d+1}\to \R^d$ that drops the last coordinate.  Thus, this lower hull, which we call the \emph{lift of $\Del(X)$}, is an embedding of $\Del(X)$ in $\R^{d+1}$.  

This embedding extends to an embedding of $\Inc(X)$ whose lower envelope is the lift of $\Del(X)$; see~\cref{fig:triangulation_lifts}. %
It is helpful to think about this embedding in terms of the lifts of consecutive pairs of Delaunay triangulations $\Del(X_i)$, $\Del(X_{i+1})$.  If $x_{i+1}$ lies in the convex hull of $X_i$, then the lift of $\Del(X_{i+1})$ partially coincides with the lift of $\Del(X_i)$ 
and otherwise lies below it in the $x_{d+1}$-direction.  The two lifts enclose a void in $\bR^{d+1}$. The upper hull of this void consists of all $d$-simplices
in $\Del(X_i)$ that do not belong to $\Del(X_{i+1})$;  the lower hull consists of the $d$-simplices in $\Del(X_{i+1})$ containing
$x_{i+1}$. Noting that $f(x_{i+1})$ lies below the hyperplane of each $d$-simplex in the upper hull of the void, one can check that the $(d+1)$-simplices of $\Inc(X_{i+1})\setminus \Inc(X_i)$ triangulate the void.

\begin{figure}[h]
  \centering
  \begin{tikzpicture}[scale=0.6]

\draw[gray, thick, domain=-4.2:4.2,opacity=.3] plot (\x, {.3*pow(\x,2});

\draw[gray,thick,< ->,,opacity=.3] (-4.5, 0) -- (4.5, 0);
\draw[black,very thick] (-3, 0) -- (4, 0);

\draw[black, fill=black] (4, 0) circle (2.5pt) node[below=4pt] {\scriptsize$x_2$};
\draw[black, fill=black] (-3, 0) circle (2.5pt) node[below=4pt] {\scriptsize$x_1$};
\draw[black, fill=black] (2, 0) circle (2.5pt) node[below=4pt] {\scriptsize$x_4$};
\draw[black, fill=black] (-1, 0) circle (2.5pt) node[below=4pt] {\scriptsize$x_3$};
\draw[black, fill=black] (1.25, 0) circle (2.5pt) node[below=4pt] {\scriptsize$x_5$};

\draw[blue, thick, fill=blue!30!white] (-3, 2.7) -- (4,4.8) -- (-1,.3)--cycle;
\draw[blue, thick, fill=blue!30!white] (-1,.3) -- (2,1.2) -- (4,4.8)--cycle;
\draw[blue, thick, fill=blue!30!white] (-1,.3) -- (1.25,.46875) -- (2,1.2) --cycle;
\draw[black,very thick] (-3, 2.7) -- (-1,.3) -- (1.25,.46875) -- (2,1.2) -- (4,4.8);

\draw[black, fill=black](4, 4.8) circle (2.5pt);
\draw[black, fill=black] (-3, 2.7) circle (2.5pt);
\draw[black, fill=black] (2, 1.2) circle (2.5pt);
\draw[black, fill=black](-1, .3) circle (2.5pt);
\draw[white, fill=white](-1.16, .26) circle (1.5pt);
\draw[black, fill=black](1.25,.46875) circle (2.5pt);

\end{tikzpicture}
    \caption{\small The Delaunay triangulation $\Del(X)$ of a set of five points $X\subset\R$ and its parabolic lift are drawn in black.  $\Inc(X)$ is drawn in blue for the order of $X$ shown. \label{fig:triangulation_lifts}}
\end{figure}
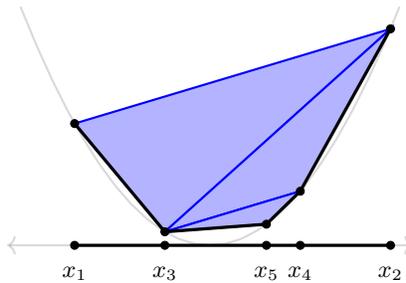

\section{Sublevel Bifiltrations}\label{sec:bifiltrations}

\subparagraph{Simplicial bifiltrations of point clouds.}
As in the introduction, let $A=[0,\infty)\times \R$, and define the sublevel Delaunay-\v{C}ech and sublevel Delaunay bifiltrations of $\gamma\colon X\to \R$, denoted $\DelCech_{\filt}(\gamma)$ and $\Del_{\filt}(\gamma)$, to be the $A$-indexed bifiltrations given by 
\begin{align*}
\DelCech_{r,s}(\gamma)&=\{\sigma \in \Inc(X) \mid \meb{\sigma}\leq r,\ \gamma(\max(\sigma))\leq s\},\\
\Del_{r,s}(\gamma)&=\{\sigma \in \Inc(X) \mid \radius{\sigma}{\gamma} \leq r,\ \gamma(\max(\sigma))\leq s\}.
\end{align*}
where $\meb{\sigma}$ denotes the radius of the minimum enclosing ball of $\sigma$ and  $\radius{\sigma}{\gamma}$ is the incremental Delaunay radius of \cref{def:incremental_delaunay}.

To establish the topological equivalence of these bifiltrations and the sublevel offset bifiltration $\offset_{\filt}(\gamma)$ defined in the introduction, we also consider the sublevel \v{C}ech bifiltration $\mathcal C_{\filt}(\gamma)$.  Recall that \v{C}ech filtration $\mathcal C_{\filt}(X)$ is the $[0,\infty)$-indexed filtration defined by
\[\mathcal C_r(X)= \Set{\sigma \subset X \given \bigcap_{x\in \sigma} B_{r}(x) \neq\varnothing},\]
where $B_{r}(x)$ denotes the closed ball of radius $r$ centered at $x$.  We then define $\mathcal C_{\filt}(\gamma)$ to be the $A$-indexed bifiltration given by $\mathcal C_{r,s}(\gamma)=\mathcal C_r(\gamma^{-1}(-\infty,s])$.  In analogy to the 1-parameter setting \cite{bauerMorseTheoryCech2016}, these bifiltrations are related by containment as follows:
\[\Del_{\filt}(\gamma)\subset \DelCech_{\filt}(\gamma)\subset \Cech_{\filt}(\gamma).\]

\subparagraph{Topological equivalence of bifiltrations.}
We next consider a notion of topological equivalence of (bi)filtrations which is
standard in classical algebraic topology and also naturally suited to TDA
\cite{corbet2023computing,blumberg2017universality,bauer2023unified,lanari2023rectification}.
We will write the definition only for $A$-indexed bifiltrations, but the definitions extend to 1-parameter filtrations
and more generally, to diagrams of topological spaces indexed by an arbitrary
partially ordered set.
Given $A$-indexed bifiltrations of topological spaces $F_{\filt}$ and $G_{\filt}$, a \emph{natural transformation} $\alpha\colon F_{\filt}\to G_{\filt}$ is a collection of
continuous maps $(\alpha_{a}\colon F_{a}\to G_{a})_{a\in A}$ such that the following diagram commutes for every $a\leq a'$:
\begin{equation*}
  \begin{tikzcd}
    F_a \arrow[r, hook]\dar[swap]{\alpha_a} & F_{a'}\dar{\alpha_{a'}}\\
    G_a  \arrow[r, hook] & G_{a'}\nospaceperiod
  \end{tikzcd}
\end{equation*}
We say that $\alpha$ is a \deff{pointwise homotopy equivalence} if each $\alpha_{a}$ is a homotopy
equivalence.  Existence of a pointwise homotopy equivalence $F_{\filt}\to G_{\filt}$ is not an equivalence relation on bifiltrations, %
but it generates one:

\begin{definition}\label{Def:Topological_Equivalence}
$F_{\filt}$ and $G_{\filt}$ are said to be \deff{topologically equivalent} (or \deff{weakly equivalent}) if they
are connected by a zigzag of pointwise homotopy equivalences:
\[
\begin{tikzcd}[ampersand replacement=\&,column sep=2ex,row sep=2ex]
   \& C^1_{\filt}\ar[swap]{dl}\ar{dr}  \&           \& \cdots\ar[swap]{dl}\ar[]{dr}  \&                \&   C^n_{\filt}\ar[swap]{dl}\ar[]{dr}  \\
F_{\filt} \&                                                             \&  C^2_{\filt} \&                                                                                                                       \& C^{n-1}_{\filt} \&                                                               \&G_{\filt}.
\end{tikzcd}
\]
\end{definition}

The following is an immediate consequence of the
(functorial) Nerve Theorem~\cite[Thm 3.9]{bauer2023unified}:

\begin{proposition}
$\mathcal C_{\filt}(\gamma)$ and $\offset_{\filt}(\gamma)$ are topologically equivalent.
\end{proposition}

To compare the topology of
$\DelCech_{\filt}(\gamma)$, $\Del_{\filt}(\gamma)$, and
$\mathcal C_{\filt}(\gamma)$, we work with \textit{collapses}, a
combinatorial notion of deformation retraction first introduced by
Whitehead~\cite{whiteheadCollapses,edelsbrunnerComputationalTopologyIntroduction2010}.
In a simplicial complex $K$, a simplex $\sigma$ is a
\deff{free face} of a simplex $\tau$ if $\tau$ is the only simplex that properly
contains $\sigma$. There is an \deff{elementary collapse} from $K$ to a
subcomplex $L$ if $K\setminus L$ is a pair of simplices
$\{\sigma,\tau\}$ such that $\sigma$ is a free face of $\tau$. We say that $K$
\deff{collapses} to $L$ and write $K\collapse L$ if there is a sequence of
elementary collapses from $K$ to $L$.  If $K\collapse L$, then
the inclusion $L \hookrightarrow K$ is a homotopy equivalence.

In~\cref{sec:proof_of_collapsing_theorem}, we prove the following theorem:
\begin{theorem}\label{thm:collapsing}
  For any $a\in [0,\infty)\times \R$,
  \begin{equation*}
    \Cech_{a}(\gamma) \collapse \Del_{a}(\gamma) \text{ and } \DelCech_{a}(\gamma) \collapse \Del_{a}(\gamma).
  \end{equation*}
\end{theorem}

As an immediate corollary, we obtain the main topological result of this paper:

  \begin{corollary}\label{thm:incremental_equiv}
    The inclusions of
    filtrations \[\Del_{\filt}(\gamma) \hookrightarrow  \DelCech_{\filt}(\gamma) \text{
      and }\Del_{\filt}(\gamma)\hookrightarrow\Cech_{\filt}(\gamma)\]
    are both pointwise homotopy equivalences.  Hence
    $\Del_{\filt}(\gamma)$, $\DelCech_{\filt}(\gamma)$, $\Cech_{\filt}(\gamma)$, and $\offset_{\filt}(\gamma)$ are all are topologically
    equivalent.
  \end{corollary}

\section{Implementation and experiments}
\label{sec:implementation}
We have written a C\texttt{++} program to compute the sublevel Delaunay-\v Cech bifiltration $\DelCech_{\filt}(\gamma)$ of a function $\gamma\colon X\to \R$, where $X\subset \R^d$ and $d\geq 2$.
It reads a text file containing $(d+1)$ float values per line; each line specifies a point in $x\in X$, with $\gamma(x)$ as the last value.
The program computes a graded chain complex representing $\DelCech_{\filt}(\gamma)$
in the \textsc{Scc2020} file format~\cite{scc2020}.
When $d=2$ or $d=3$, the program uses \textsc{CGAL}'s specialized Delaunay triangulations for these dimensions \cite{cgal_dt_2d,cgal_dt_3d}; when $d>3$, it uses CGAL's implementation of Delaunay triangulations for arbitrary dimensions~\cite{cgal_dt_dd}.

After augmenting $X$ to enforce the $\Delta$-property as in \cref{rem:Delta_Prop}, our program computes all $(d+1)$-simplices of the incremental Delaunay complex using the variant of the Bowyer-Watson algorithm described in \cref{sec:incr_del_cmplx}.  These simplices are inserted into a \emph{simplex tree}~\cite{simplex_tree_paper}, an efficient data structure for simplicial complexes available in \textsc{Gudhi}~\cite{gudhi:FilteredComplexes}.
The meb radii of all simplices are then computed using \textsc{CGAL}~\cite{cgal_meb} (for $d\leq 10$)
or the software package \textsc{Miniball} by G{\"a}rtner\footnote{\url{https://people.inf.ethz.ch/gaertner/subdir/software/miniball.html}} (for $d>10$).  For efficiency, we compute the meb radii in order of decreasing simplex dimension, exploiting the fact that if the meb of $\sigma$ contains $\tau\subset \sigma$ in its interior, then the same ball is also the meb of $\sigma\setminus \tau$.
Finally, the simplices are ordered with respect to their meb radius, the function value of each simplex is determined, and the output file is written.

Thanks to its use of the aforementioned software packages, our program consists of only around 1500 lines of code.
We offer it in a public repository\footnote{\url{https://bitbucket.org/mkerber/function_delaunay}} published under the GNU Public License (GPL).

All experiments were performed on a computer with an
Intel Core i7-5960X CPU @ 3.00GHz and
 64GB of memory,
running Ubuntu 20.04.6 LTS.
The program was compiled with \texttt{g++ 9.4.0}.

\subparagraph{Datasets.}
We have generated various point clouds by sampling $1$-spheres $S^1$ and unit squares $[0, 1]^2$ in $\R^2$, and $2$-spheres $S^2$, tori $S^1\times S^1$ and unit cubes $[0, 1]^3$ in $\R^3$.
For each point cloud, we add 5\% of noise, sampled from the uniform distribution.
In addition, a small perturbation (5\%) is introduced to each point cloud to ensure that point clouds are in general position.
For each of these five data types, we considered samples of $500$, $1000$, $2000$, $4000$, $8000$, $16000$ and $32000$ points.
For fixed type and number of points, we generated five independent samples, resulting in $125$ different point clouds.

For each point cloud, we tested the following four choices of functions %
\begin{itemize}
    \item \emph{codensity function} $\gamma(p) = -\sum\limits_{q\neq p}{\exp{\left(-\frac{\|p-q\|^2)}{\sigma^2}\right)}}$, with $\sigma$ chosen as $0.1^{\mathrm{st}}$ percentile of the non-zero distances between points in $X$.
    \item \emph{$L_1$-coeccentricity} $\rm{ecc}_1 (p) = -\sum\limits_{q\in X}{\frac{\|p-q\|}{|X|}}$, where $|X|$ denotes the size of $X$.
    \item \emph{height} $h(p) = p_d$, where $p_d$ is the last coordinate of $p\in \R^d$
    \item \emph{random}, where each points gets an independent, random value in $[0,1]$.
\end{itemize}
In total, our test suite contains a total of $500$ point clouds with functions. For brevity, in what follows we only present a subset of representative results.

\subparagraph{Complex size.}
To get a sense of how large incremental Delaunay triangulations are in practice, we experimentally compare the sizes of incremental Delaunay complexes and Delaunay triangulations.   Table~\ref{table:size} gives such a comparison
for points on the circle in $\R^2$ and on the torus in $\R^3$.
There are two major observations: first, for randomly
assigned function values, the ratio of the sizes remains (nearly) constant. That is expected, as this case corresponds
to the randomized incremental construction (RIC) paradigm, and it is known that the size of the RIC is small in expectation.
The second observation is that the non-random functions lead to bigger incremental Delaunay complex, and moreover, the ratio of the sizes is growing as the number of points grows.
While we cannot offer a theoretical explanation of this behavior, we suspect that the insertion order
of the points is particularly disadvantageous for the incremental construction: for instance, in the height function,
every newly inserted point is outside of the convex hull of the previous points which is not typical for RIC.
Nevertheless, we observe that the ratio grows slowly and is rather small, making the computation of the
incremental Delaunay complex still feasible for large point clouds in small dimension.

\begin{table}[h]
  \resizebox{\textwidth}{!}{
  \begin{tabular}{ccc||cc| cc | cc | cc }
    \hline
    \multirow{2}{*}{Sample} & \multirow{2}{*}{\#Points} & \multirow{1}{*}{Delaunay}& \multicolumn{2}{c}{Random} & \multicolumn{2}{c}{Density} & \multicolumn{2}{c}{Height} & \multicolumn{2}{c}{Eccentricity} \\
    & & size & Size & Ratio & Size & Ratio & Size & Ratio & Size & Ratio \\

\hline
\multirow{6}{*}{$S^1$}              & 1000 & 5975 & 19481 & 3.26 & 23611 & 3.95 & 25575 & 4.28 & 39419 & 6.6 \\
                                    & 2000 & 11971 & 39131 & 3.27 & 49593 & 4.14 & 57019 & 4.76 & 83209 & 6.95 \\
                                    & 4000 & 23959 & 79839 & 3.33 & 102307 & 4.27 & 127111 & 5.31 & 180307 & 7.53 \\
                                    & 8000 & 47967 & 158789 & 3.31 & 209099 & 4.36 & 282055 & 5.88 & 375891 & 7.84 \\
                                    & 16000 & 95953 & 318919 & 3.32 & 430991 & 4.49 & 619561 & 6.46 & 788465 & 8.22 \\
                                    & 32000 & 191943 & 638553 & 3.33 & 889599 & 4.63 & 1343567 & 7.0 & 1638579 & 8.54 \\

\hline
\multirow{6}{*}{$S^1 \times S^1$}   & 1000 & 27003 & 122181 & 4.52 & 171105 & 6.34 & 230075 & 8.52 & 215815 & 7.99 \\
                                    & 2000 & 52717 & 257129 & 4.88 & 361151 & 6.85 & 515437 & 9.78 & 467803 & 8.87 \\
                                    & 4000 & 105249 & 512521 & 4.87 & 747617 & 7.1 & 1159641 & 11.02 & 937369 & 8.91 \\
                                    & 8000 & 211885 & 1010347 & 4.77 & 1526509 & 7.2 & 2561545 & 12.09 & 1911169 & 9.02 \\
                                    & 16000 & 426531 & 2028761 & 4.76 & 3093869 & 7.25 & 5578187 & 13.08 & 3912707 & 9.17 \\
                                    & 32000 & 862789 & 4065635 & 4.71 & 6376129 & 7.39 & 12145149 & 14.08 & 8118607 & 9.41 \\
\hline

\end{tabular}}
\centering
  \caption{Size comparison of incremental Delaunay complexes and Delaunay triangulations.}
  \label{table:size}
\end{table}

\subparagraph{Computation time.}
We report on the time and memory to compute the sublevel Delaunay bifiltration (in chain complex format)
in Table~\ref{table:time_overview}. Time for file IO is not reported.
We observe a slightly super-linear performance for all function choices. That is expected from the previous
experiments for the non-random functions since the size of the output complex grows super-linear as well, but remarkably,
also applies to random function values, although the growth seems slightly smaller.
We also normalized the runtime by dividing it by the complex size, and we observe that the time spent
per simplex tends to grow for larger instances. Nevertheless, despite an apparent super-linear behavior
with respect to the output size, the program can compute bifiltrations for tens of thousands of points
in $\R^3$ within less than a minute. Also, the memory consumption is close to linear.

\begin{table}[h]
  \resizebox{\textwidth}{!}{
  \begin{tabular}{cc||ccc|ccc|ccc|ccc}
    \hline
    \multirow{3}{*}{Sample} & \multirow{3}{*}{\#Points} & \multicolumn{3}{c}{Random} & \multicolumn{3}{c}{Density} & \multicolumn{3}{c}{Height} & \multicolumn{3}{c}{Eccentricity} \\
    & & Time &  Time/simplex  & Memory & Time & Time/simplex & Memory & Time & Time/simplex & Memory & Time & Time/simplex & Memory \\
    & & (s) &   (ms) & (GB) & (s) & (ms) & (MB) & (s) & (ms) & (MB) & (s) & (ms) & (MB) \\
\hline
\multirow{6}{*}{$S^1$}              &  1000 & 0.03 & 1.58 & 8.44 & 0.05 & 2.01 & 9.4 & 0.04 & 1.5 & 9.72 & 0.05 & 1.35 & 13.12 \\
                                    &  2000 & 0.06 & 1.62 & 12.87 & 0.08 & 1.54 & 15.08 & 0.08 & 1.33 & 17.14 & 0.14 & 1.67 & 22.79 \\
                                    &  4000 & 0.13 & 1.69 & 21.88 & 0.16 & 1.58 & 26.48 & 0.17 & 1.36 & 32.86 & 0.3 & 1.66 & 44.19 \\
                                    &  8000 & 0.28 & 1.78 & 39.56 & 0.35 & 1.68 & 50.13 & 0.44 & 1.55 & 68.38 & 0.64 & 1.71 & 87.75 \\
                                    &  16000 & 0.64 & 2.02 & 74.72 & 0.81 & 1.87 & 98.46 & 1.03 & 1.66 & 146.57 & 1.44 & 1.82 & 178.85 \\
                                    &  32000 & 1.5 & 2.35 & 147.81 & 1.79 & 2.01 & 199.02 & 2.37 & 1.76 & 311.03 & 3.17 & 1.94 & 371.64 \\

\hline
\multirow{6}{*}{$S^1 \times S^1$}   &  1000 & 0.21 & 1.73 & 31.89 & 0.3 & 1.77 & 43.03 & 0.41 & 1.79 & 57.21 & 0.38 & 1.76 & 53.09 \\
                                    &  2000 & 0.49 & 1.92 & 62.76 & 0.75 & 2.06 & 85.69 & 1.05 & 2.03 & 125.26 & 0.95 & 2.04 & 110.82 \\
                                    &  4000 & 1.11 & 2.17 & 122.72 & 1.59 & 2.12 & 170.57 & 2.44 & 2.1 & 265.98 & 1.98 & 2.11 & 224.83 \\
                                    &  8000 & 2.36 & 2.34 & 236.03 & 3.49 & 2.29 & 354.77 & 5.54 & 2.16 & 619.1 & 4.17 & 2.18 & 437.08 \\
                                    &  16000 & 5.11 & 2.52 & 455.98 & 7.41 & 2.4 & 710.66 & 12.7 & 2.28 & 1271.56 & 8.94 & 2.28 & 910.05 \\
                                    &  32000 & 11.1 & 2.73 & 929.58 & 15.65 & 2.45 & 1461.12 & 28.5 & 2.35 & 2796.81 & 19.08 & 2.35 & 1847.01 \\

 \hline
\end{tabular}}
\centering
  \caption{Running time (in seconds), the average time to process a simplex (in millisecond), and the memory consumption (in MB) for computing the sublevel Delaunay bifiltration.}
  \label{table:time_overview}
\end{table}

\subparagraph{Comparison with sublevel Rips.}
We next compare the performance of homology computations using sublevel Delaunay-\v Cech and sublevel Rips bifiltrations.
Still fixing a point set $X$ and a function $\gamma:X\to\R$, the sublevel Rips bifiltration $\rips_{\filt}(\gamma)$ is defined by taking $\rips_{r,s}(\gamma)$ to consist of all cliques $S\subset X$ of diameter at most $r$ whose points have $\gamma$-value at most $s$.  As noted in the introduction, because the full sublevel Rips bifiltration $\rips_{\filt}(\gamma)$ is large, we usually consider its $k$-skeleton for $k$ small. Here we will take
$k=2$, which is sufficient for homology computations of degree at most one.
Still, the $2$-skeleton of $\rips_{\filt}(\gamma)$ has a cubic number of simplices in $|X|$, which is prohibitively large
for the data sizes we consider.

However, recent work has introduced way to compute a smaller but equivalent bifiltration: Note that for $r$ and $s$ large enough, $\rips_{r,s}(\gamma)$
is a clique complex
over a complete graph. Alonso, Kerber, and Pritam~\cite{alonso2023filtration}
preprocess the complete graph by removing edges to arrive at a smaller graph $G$, and show that the bifiltration $\mathcal G_{\filt}(\gamma)$ induced
by the clique complex of $G$ is equivalent to $\rips_{\filt}(\gamma)$.

Given a simplicial filtration $\mathcal Z$, the software \emph{mpfree}\footnote{\url{https://bitbucket.org/mkerber/mpfree/src/master/}}~\cite{lw-computing,kr-fast,fkr-compression}
efficiently computes a minimal presentation of the homology of $\mathcal Z$ in a fixed degree.

We use \emph{mpfree} and the code from~\cite{alonso2023filtration}\footnote{\url{https://github.com/aj-alonso/filtration_domination}}
to compute (a minimal presentation of) the homology of $\rips_{\filt}(\gamma)$ in degree 1.  We also use
\emph{mpfree} to compute the homology of $\Inc(\gamma)$ in degree 1.  These computations require us only to compute the 2-skeleta of $\mathcal G_{\filt}(\gamma)$ and $\Inc_{\filt}(\gamma)$.
We note that the results of the two computations are \emph{not} equivalent, because $\rips_{\filt}(\gamma)$
is not topologically equivalent to $\offset_{\filt}(\gamma)$.  However, as indicated in the introduction, the two bifiltrations approximate each other, and hence the same is true for their homology~\cite{blumberg2017universality}.

Table~\ref{table:rips_compare} shows the results of our comparison. We see that our approach improves
on the Rips variant in terms of complex size, time, and memory, usually by several orders of magnitude.
This is the case even though the approach of~\cite{alonso2023filtration}
substantially reduces the number of edges in the bifiltration and therefore vastly improves on a naive computational strategy \cite{alonso2023filtration}.

The software
\emph{2pac}\footnote{\url{https://gitlab.com/flenzen/2-parameter-persistent-cohomology}},
based on recent work of Bauer, Lenzen, and Lesnick \cite{bauer2023efficient}, is
an alternative to \emph{mpfree} which is optimized for clique bifiltrations.  In this
work, we have not run experiments with \emph{2pac}.  While it is possible that
using the code from~\cite{alonso2023filtration} together with \emph{2pac} would
lead to significant performance gains, for low-dimensional point clouds this
seems unlikely to make the performance of $\rips_{\filt}(\gamma)$ competitive with $\Inc_{\filt}(\gamma)$.
\begin{table}[h]
\centering
\resizebox{\textwidth}{!}{

\begin{tabular}{cc |ccccc|ccc}
\hline
\multirow{4}{*}{Data Set} & \multirow{4}{*}{\#Points} & \multicolumn{5}{c|}{\multirow{2}{*}{codensity-Rips bifiltration }} & \multicolumn{3}{c}{\multirow{3}{*}{codensity-Delaunay bifiltration}} \\

 &  & \multicolumn{5}{c|}{\multirow{2}{*}{preprocessing with filtration-domination }} & \multicolumn{3}{c}{} \\
  &  && & & & & & & \\
 \cline{3-10}
 &  & {\# Edges before } & {\# Edges after} & Complex Size & Time & Memory & Complex Size & Time & Memory \\

\hline
\hline
\multirow{2}{*}{$S^1$}              &  500 & 124750 & 15124 & 10293094 & 16.24 & 2573.28 & 11269 & 0.02 & 6.78 \\
                                    &  1000 & 499500 & 44219 & 86436503 & 137.04 & 21778.87 & 23611 & 0.06 & 9.82 \\
\multirow{2}{*}{$S^2$}              &  500 & 124750 & 50342 & 188414574 & 246.45 & 47241.98 & 67763 & 0.13 & 19.49 \\
                                    &  1000 & 499500 & 173725 & $\infty$ & $\infty$ & $\infty$ & 150747 & 0.32 & 38.17 \\
\multirow{2}{*}{$[0, 1]^2$}         &  500 & 124750 & 5922 & 136392 & 2.36 & 37.09 & 11321 & 0.03 & 7.89 \\
                                    &  1000 & 499500 & 17521 & 965123 & 16.76 & 242.81 & 23497 & 0.08 & 12.8 \\
\multirow{2}{*}{$[0, 1]^3$}         &  500 & 124750 & 10717 & 1134028 & 3.55 & 281.9 & 70975 & 0.17 & 20.65 \\
                                    &  1000 & 499500 & 43310 & 20634286 & 46.91 & 5200.04 & 159887 & 0.36 & 40.43 \\
\multirow{2}{*}{$S^1 \times S^1$}   &  500 & 124750 & 16665 & 2775487 & 5.45 & 689.69 & 73967 & 0.14 & 21.09 \\
                                    &  1000 & 499500 & 65558 & 46094578 & 73.07 & 11486.7 & 171105 & 0.36 & 43.69 \\
\hline
\end{tabular}}
\caption{Running time (in seconds) and memory consumption (in MB) of minimal presentation computation in homology degree 1.  The symbol $\infty$ means the pipeline ran out of memory and the computation could not finish.}\label{table:rips_compare}
\end{table}

\subparagraph{Further experiments.}
We summarize more experimental findings. The details are in Appendix~\ref{app:full_experiments}.

\textbf{Subroutines:} We measured the performance of the major substeps of our bifiltration algorithm.
It turns out that only around $1/6$ of the runtime is required to construct the incremental Delaunay complex,
around $1/3$ for computing the minimum enclosing ball radii, and $1/2$ for creating the chain complex structure
that we output.

\textbf{Higher dimensions:} Our software is not limited to points clouds in $2$ and $3$ dimensions. We ran some experiments for higher dimensional data. Not surprisingly, complex size,
time, and memory increase significantly when the ambient dimension increases. However, the ratio between the Delaunay triangulation and our incremental construction remains quite small ($\leq 10$) in the instances we tested.

\textbf{Compression:} Given the chain complex of a sublevel Delaunay bifiltration, the software \emph{multi-chunk}\footnote{Available at \url{https://bitbucket.org/mkerber/multi_chunk/src/master/}}~\cite{fk-chunk,fkr-compression}, computes an algebraically equivalent chain complex whose size is substantially smaller---by up to a factor of 30, in our experiments.

\section{Conclusion}
\label{sec:conclusion}
We have introduced the Delaunay-\v Cech bifiltration $\DelCech_{\filt}(\gamma)$ of a function $\gamma$ on a Euclidean point cloud.  We have seen that
\begin{enumerate}
\item $\DelCech_{\filt}(\gamma)$ is topologically well behaved, in the sense that it is equivalent to the offset bifiltration,
\item $\DelCech_{\filt}(\gamma)$ is reasonably small, both in theory and in practice, and
\item for low-dimensional data, $\DelCech_{\filt}(\gamma)$ is far more computable than any known alternative.
\end{enumerate}
 Our construction enables bipersistence analysis of tens of thousands of points in low dimensions.  Since such data is common in applications and is already actively studied using 2-parameter persistence  \cite{keller2018persistent,vipond2021multiparameter,benjamin2022multiscale}, we foresee our algorithm and its implementation having an immediate practical impact.  In fact, for Lipschitz functions $\gamma$, it should be possible to handle far bigger data sets by subsampling, at the cost of some small approximation error in the interleaving distance \cite{lesnick2015theory}.%

We have also introduced a second, topologically equivalent bifiltration, the sublevel Delaunay bifiltration $\Del_{\filt}(\gamma)$, a sub-bifiltration of $\DelCech_{\filt}(\gamma)$ differing only in the radii at which simplices appear.  It is natural to ask which radii are faster to compute.  For $\DelCech_{\filt}(\gamma)$, the birth radii are obtained via minimum enclosing ball computations, for which high-quality code is available. 
It is also possible to compute the birth radii for $\Del_{\filt}(\gamma)$ in constant time per simplex, by extending the usual approach to computing birth radii for the 1-parameter Delaunay filtration \cite[Chapter 6.1.2]{boissonnat2018geometric}; the details will be given in a future version of this paper.
As we have not implemented this, it is unclear which approach would be faster in practice.  It may also be that we can speed up our computations of $\DelCech_{\filt}(\gamma)$ by using a specialized algorithm to compute minimum enclosing balls.  %

We believe that Delaunay triangulations could be useful for multiparameter persistence in ways that extend well beyond our work.
A natural first question is how to extend our approach to the case of multiple functions $\gamma_1,\ldots,\gamma_m:X\to \R$.
A solution could be useful in the study of dynamic point processes, which can be modeled using a
pair of functions tracking the entrance and removal time of points.
A similar problem is to identify a Delaunay variant of the \emph{degree-Rips bifiltration}, an actively studied density-sensitive bifiltration whose definition requires no choice of bandwidth parameter~\cite{lesnick2015interactive,blumberg2022stability,rolle-degree}.

\subparagraph{Acknowledgements.} We thank the anonymous reviewers for their
valuable comments and suggestions. ÁJA and MK acknowledge the support of the
Austrian Science Fund (FWF) grant P 33765-N.
ML was supported by a grant from the Simons Foundation (Award ID 963845).
MK and ML acknowledge the Dagstuhl Seminar 23192 “Topological Data
Analysis and Applications” that initiated this collaboration.

\begin{appendix}

\section{Proof of \texorpdfstring{\cref{thm:collapsing}}{Theorem \ref{thm:collapsing}}}
\label{sec:proof_of_collapsing_theorem}

We briefly recall the concepts from discrete Morse theory~\cite{formanMorseTheoryCell1998,chariDiscreteMorseFunctions2000} needed for our argument:
 A \deff{
 (discrete) vector field} $V$ on an abstract simplicial complex $K$ is a partition of a subset of $K$ into pairs
$(\sigma,\tau)$ with $\sigma$ a facet of $\tau$. Consider the Hasse
diagram of $K$ as an acyclic directed graph $G$ with edges pointing from
smaller to larger faces.  Formally, each element of $V$ is an edge in $G$.  If the graph obtained from $G$ by reversing the direction of all edges in $V$ is acyclic, then we call
$V$ a \deff{gradient}.  %

 We are interested in gradients because
of the following collapsing theorem~\cite{formanMorseTheoryCell1998}:
\begin{theorem}\label{thm:morse_collapse}
For simplicial complexes $K'\subset K$, if there is a
  gradient $V$ on $K$ such that $K\setminus K'$ is the union of the pairs of $V$, then
  $K\collapse K'$.
\end{theorem}

We omit the easy proofs of the following two lemmas:
\begin{lemma}\label{lem:subDGVF}
For simplicial complexes $K'\subset K$, let $V$ be a gradient on $K$, and let $V'$ be a vector field on $K'$ such that each pair in $V'$ is also in $V$.  Then $V'$ is a gradient on $K'$.
\end{lemma}

\begin{lemma}[{\cite[Lemma
5.1]{bauerMorseTheoryCech2016}}]\label{lem:vertex_gradient}
  Let $V$ be a vector field on a simplicial complex $K$, and fix a vertex
  $x$ of $K$. If all pairs of $V$ are of the form $(\sigma,\sigma \cup \{x\})$,
  then $V$ is a gradient.
\end{lemma}

To prove \cref{thm:collapsing}, we will use the language of \emph{selective Delaunay complexes} from~\cite{bauerMorseTheoryCech2016}. Given a subset $E$ of
a finite point set $Y\subset\bR^{d}$, the selective Delaunay complex
$\SelDel(Y, E)$
consists of all
simplices $\sigma \subset Y$ for which there exists a $(d-1)$-sphere $S$ such
that
\begin{enumerate}
\item all points of $\sigma$ lie on or inside $S$ and
\item no point of $E$ lies inside $S$.
\end{enumerate}
For $\sigma \in \SelDel(Y, E)$, let $\selradius{\sigma}{E}$ denote the minimum
radius of such $S$, and let \[\SelDel_{r}(Y, E)=\{\sigma\in\SelDel(Y, E) \mid \selradius{\sigma}{E}\leq r\}.\]

Note that %
$\SelDel_r(Y, \varnothing)=\mathcal C_r(Y)$, and that if $E\subset F\subset Y$, then $\SelDel_{r}(Y, F) \subset \SelDel_{r}(Y, E)$.
The sublevel Delaunay bifiltration $\Del_{\filt}(\gamma)$ can be described in terms of selective Delaunay complexes, as follows:
A subset $\sigma\subset X$ with $x_{k} \coloneqq \max\sigma$ is a simplex in $\Del_{r,s}(\gamma)$ if and only if
$\gamma(x_{k}) \leq s$ and $\sigma\in\Del_{r}(X_{k}, X_{k-1})$. In particular, if $\sigma\in \Inc(X)$, then $\rho_\sigma = \selradius{\sigma}{X_{k-1}}$.

As in~\cite{bauerMorseTheoryCech2016}, for $\sigma \subset X$ and
$x\in X$, we let $\sigma - x$ and $\sigma + x$ denote the sets
$\sigma\setminus\Set{x}$ and $\sigma\cup\Set{x}$, respectively. Note that one of
the sets $\sigma-x$, $\sigma+x$ is equal to $\sigma$.

\subparagraph{From sublevel \v{C}ech to sublevel Delaunay.} We first prove the following:
\begin{theorem}\label{thm:cech_to_sublevel_del}
  For any $a\in [0, \infty) \times \bR$,
  $\Cech_{a}(\gamma)\collapse \Del_{a}(\gamma)$.
\end{theorem}

Writing $a=(r,s)$, let $j=\max\,\{i\mid \gamma(x_i)\leq s\}$.   %
For $i\in\{0,\ldots,j\}$, let $\mathcal U_i = \Del_{r}(X_j,X_i) \cup \Del_{a}(\gamma)$, where we adopt the convention that $X_0=\varnothing$.  Note that we have a filtration
\begin{equation*}
\Del_{a}(\gamma)= \mathcal U_j\,\subset\, \mathcal U_{j-1}\,\subset\,\cdots\,\subset\, \mathcal U_0= \Cech_{a}(\gamma).
\end{equation*}
To prove \cref{thm:cech_to_sublevel_del}, it suffices to show that $\mathcal{U}_{i-1} \collapse \mathcal{U}_i$ for each $i\in \{1,\ldots,j\}$.  We will use the next lemma, which follows immediately
from~\cite[Lemma 5.5]{bauerMorseTheoryCech2016} and~\cite[Lemma
5.7]{bauerMorseTheoryCech2016}.
\begin{lemma}\label{lem:first_pairing}
For $i\in \{1,\ldots, j\}$ and $\sigma \in \Del_{r}(X_j, X_{i-1})\setminus \Del_{r}(X_j, X_{i})$, we also  have $\sigma\pm x_{i}\in\Del_{r}(X_j, X_{i-1})\setminus \Del_{r}(X_j, X_{i})$.
\end{lemma}
This implies the analogous statement for the complexes $\mathcal U_i$:
\begin{lemma}\label{lem:pairing}
For $i\in \{1,\ldots, j\}$ and $\sigma\in\mathcal U_{i-1}\setminus \mathcal U_{i}$, we have
 $\sigma\pm x_{i}\in\mathcal U_{i-1}\setminus \mathcal U_{i}$.
\end{lemma}
\begin{proof}
We need to show that if $\sigma\in \SelDel_{r}(X_j, X_{i-1}) \setminus \SelDel_{r}(X_j, X_{i})$ and
  $\sigma\not\in\Del_{a}(\gamma)$, then the same holds for $\sigma\pm x_{i}$.    \cref{lem:first_pairing} gives that
  $\sigma\pm x_{i}\in\SelDel_{r}(X_j, X_{i-1})\setminus\SelDel_{r}(X_j, X_{i})$, so it remains only to show that $\sigma\pm x_{i}\not\in\Del_{a}(\gamma)$.
 Let  $x_{k}\coloneqq \max\sigma$. We claim that $x_{k} > x_{i}$. Suppose to
  the contrary that $x_{k}\leq x_{i}$. Then $X_{k-1}\subset X_{i-1}$ and since
  $\sigma\in\Del_{r}(X_{k}, X_{i-1})$, it follows that
  $\sigma\in\Del_{r}(X_{k}, X_{k-1})$. But $\Del_{r}(X_{k}, X_{k-1})\subset \Del_{a}(\gamma)$, so $\sigma\in\Del_{a}(\gamma)$, a contradiction. We conclude that $x_{k} > x_{i}$ and
  $\max(\sigma\pm x_{i}) = x_{k}$.  Since
  $\sigma\pm x_{i}\not\in \Del_{r}(X_{j}, X_{i})$, we also have
  $\sigma\pm x_{i}\not\in\Del_{r}(X_{k}, X_{k-1})$, and consequently,
  $\sigma\pm x_{i}\not\in\Del_{a}(\gamma)$.
\end{proof}

\begin{proof}[Proof of \cref{thm:cech_to_sublevel_del}]
  For each $i \in \Set{1,\dots,j}$, \cref{lem:pairing} gives a partition of
  $\mathcal U_{i-1}\setminus \mathcal U_{i}$ into pairs of the form
  $(\sigma-x_{i},\sigma +x_{i})$. By~\cref{lem:vertex_gradient}, this is a
  gradient.~\cref{thm:morse_collapse} now implies that
  $\mathcal{U}_{i-1} \collapse \mathcal{U}_i$.
\end{proof}

\subparagraph{From sublevel Delaunay-\v{C}ech to sublevel Delaunay.} The
following finishes the proof of~\cref{thm:collapsing}:
\begin{theorem}\label{thm:collapse_sublevel_delcech_del}
  For any $a\in [0,\infty)\times\bR$,
  $\DelCech_{a}(\gamma)\collapse \Del_{a}(\gamma)$.
\end{theorem}

For $i\in\Set{1, \dots, \abs{X}}$, let
\begin{align*}
X_{\leq i} &= \Set{\sigma\subset X \given \max \sigma \leq x_{i}},\\
X_{=i} &= \Set{\sigma\subset X \given \max \sigma = x_{i}},\\
X_{>i} &= \Set{\sigma\subset X \given \max \sigma > x_{i}}.
\end{align*}
Let $X_{\leq 0} = X_{=0} = \varnothing$ and
$X_{>0} = X$.  For $i\in \Set{0, \dots, \abs{X}}$, define the simplicial complex
\begin{equation*}
  M_{i}= (X_{>i} \cap \Del_{a}(\gamma)) \cup (X_{\leq i} \cap
\DelCech_{a}(\gamma)),
\end{equation*}
which is indeed simplicial because, for
$\sigma\subset\sigma'\in\Del_{a}(\gamma)$, we have
$r_\sigma \leq\rho_\sigma\leq\rho_{\sigma'}$. As before, write $a=(r,s)$ and let
$j=\max\,\{i\mid \gamma(x_i)\leq s\}$. Note that we have a filtration
\begin{equation*}
  \Del_{a}(\gamma) = M_{0} \,\subset\, M_{1} \,\subset\, \dots \,\subset\, M_{j-1} \,\subset\, M_{j} = \DelCech_{a}(\gamma).
\end{equation*}
As above, to prove~\cref{thm:collapse_sublevel_delcech_del},
it suffices to show that  $M_{i} \collapse M_{i-1}$ for each $i\in\Set{1,\dots,j}$.  We use the following result from~\cite{bauerMorseTheoryCech2016} about selective Delaunay
complexes.

\begin{proposition}\label{thm:selective_delcech_to_del}
  For any $i\in\Set{1,\dots, \abs{X}}$, there is a gradient $V_i$ on
  $\Cech_{r}(X_{i})\cap\SelDel(X_{i}, X_{i-1})$ that induces a
  collapse to $\SelDel_{r}(X_{i}, X_{i-1})$, and whose pairs $(\sigma,\sigma+x)$
  satisfy $x< x_{i}$.
\end{proposition}

\begin{proof}
The vector field $V_i$ is defined by~\cite[Equation 30 and Lemma 5.6]{bauerMorseTheoryCech2016}; it is shown in the proof of~\cite[Theorem 5.9]{bauerMorseTheoryCech2016} that $V_i$ is indeed a gradient.
\end{proof}

\begin{proof}[Proof of \cref{thm:collapse_sublevel_delcech_del}]
Note that \[M_{i}\setminus M_{i-1}=X_{=i}\cap(\DelCech_{a}(\gamma)\setminus \Del_{a}(\gamma)).\]  Thus, for
 a simplex $\sigma\in\DelCech_{a}(\gamma)$, we have  $\sigma\in M_{i}\setminus M_{i-1}$ if and only if
  $\max \sigma = x_{i}$ and $r_\sigma \leq r < \rho_\sigma$.  It follows that $M_{i}\setminus M_{i-1}=X_{=i}\cap N_{i}$, where
  \[N_{i} \coloneqq \Cech_{r}(X_{i})\cap(\SelDel(X_{i}, X_{i-1})\setminus \SelDel_{r}(X_{i}, X_{i-1})).\]

  $V_i$ partitions $N_{i}$ into pairs of the form $(\sigma,\sigma+x)$ with
  $x< x_{i}$, so if $\sigma\in X_{=i}$, then $\sigma+ x\in X_{=i}$. Thus,
  $V_{i}$ restricts to a partition of $X_{=i}\cap N_{i}$, which is a
  gradient on $M_{i}$ by \cref{lem:subDGVF}.   \cref{thm:morse_collapse}~now yields
  the desired collapse.
\end{proof}

\section{Further experimental evaluation}
\label{app:full_experiments}
\subparagraph{Full evaluation of size and performance.}
We enlist the comprehensive results from our experimental runs in Tables~\ref{table:size_full} and~\ref{table:time_overview_full}.

\begin{table}[h]
  \resizebox{\textwidth}{!}{
  \begin{tabular}{ccc||cc| cc | cc | cc }
    \hline
    \multirow{2}{*}{Sample} & \multirow{2}{*}{\#Points} & \multirow{1}{*}{Delaunay}& \multicolumn{2}{c}{Random} & \multicolumn{2}{c}{Density} & \multicolumn{2}{c}{Height} & \multicolumn{2}{c}{Eccentricity} \\
    & & size & Size & Ratio & Size & Ratio & Size & Ratio & Size & Ratio \\
    \hline
\multirow{6}{*}{$[0, 1]^2$}         & 1000 & 5971 & 19223 & 3.22 & 23497 & 3.94 & 28167 & 4.72 & 34807 & 5.83 \\
                                    & 2000 & 11971 & 39801 & 3.32 & 48605 & 4.06 & 62815 & 5.25 & 74391 & 6.21 \\
                                    & 4000 & 23969 & 78859 & 3.29 & 99339 & 4.14 & 138341 & 5.77 & 155301 & 6.48 \\
                                    & 8000 & 47963 & 158443 & 3.3 & 201929 & 4.21 & 304331 & 6.35 & 326321 & 6.8 \\
                                    & 16000 & 95951 & 319753 & 3.33 & 410575 & 4.28 & 660215 & 6.88 & 682773 & 7.12 \\
                                    & 32000 & 191959 & 639235 & 3.33 & 838631 & 4.37 & 1430915 & 7.45 & 1408505 & 7.34 \\

\hline
\multirow{6}{*}{$[0, 1]^3$}         & 1000 & 27609 & 122177 & 4.43 & 159887 & 5.79 & 198433 & 7.19 & 244257 & 8.85 \\
                                    & 2000 & 56299 & 255927 & 4.55 & 345469 & 6.14 & 458447 & 8.14 & 551599 & 9.8 \\
                                    & 4000 & 113975 & 529063 & 4.64 & 731093 & 6.41 & 1050155 & 9.21 & 1193795 & 10.47 \\
                                    & 8000 & 227809 & 1070993 & 4.7 & 1518117 & 6.66 & 2376689 & 10.43 & 2562581 & 11.25 \\
                                    & 16000 & 458947 & 2172081 & 4.73 & 3166425 & 6.9 & 5281167 & 11.51 & 5418085 & 11.81 \\
                                    & 32000 & 921399 & 4359419 & 4.73 & 6607583 & 7.17 & 11650367 & 12.64 & 11290949 & 12.25 \\

\hline
\multirow{6}{*}{$S^1$}              & 1000 & 5975 & 19481 & 3.26 & 23611 & 3.95 & 25575 & 4.28 & 39419 & 6.6 \\
                                    & 2000 & 11971 & 39131 & 3.27 & 49593 & 4.14 & 57019 & 4.76 & 83209 & 6.95 \\
                                    & 4000 & 23959 & 79839 & 3.33 & 102307 & 4.27 & 127111 & 5.31 & 180307 & 7.53 \\
                                    & 8000 & 47967 & 158789 & 3.31 & 209099 & 4.36 & 282055 & 5.88 & 375891 & 7.84 \\
                                    & 16000 & 95953 & 318919 & 3.32 & 430991 & 4.49 & 619561 & 6.46 & 788465 & 8.22 \\
                                    & 32000 & 191943 & 638553 & 3.33 & 889599 & 4.63 & 1343567 & 7.0 & 1638579 & 8.54 \\

\hline
\multirow{6}{*}{$S^2$}              & 1000 & 26439 & 116635 & 4.41 & 150747 & 5.7 & 175995 & 6.66 & 263183 & 9.95 \\
                                    & 2000 & 54327 & 245651 & 4.52 & 325791 & 6.0 & 415555 & 7.65 & 600589 & 11.06 \\
                                    & 4000 & 109987 & 503831 & 4.58 & 709303 & 6.45 & 962039 & 8.75 & 1347043 & 12.25 \\
                                    & 8000 & 222339 & 1040769 & 4.68 & 1534939 & 6.9 & 2165395 & 9.74 & 2984415 & 13.42 \\
                                    & 16000 & 450133 & 2104495 & 4.68 & 3241029 & 7.2 & 4850005 & 10.77 & 6380329 & 14.17 \\
                                    & 32000 & 909567 & 4266411 & 4.69 & 6824453 & 7.5 & 10761403 & 11.83 & 13738989 & 15.1 \\

\hline
\multirow{6}{*}{$S^1 \times S^1$}   & 1000 & 27003 & 122181 & 4.52 & 171105 & 6.34 & 230075 & 8.52 & 215815 & 7.99 \\
                                    & 2000 & 52717 & 257129 & 4.88 & 361151 & 6.85 & 515437 & 9.78 & 467803 & 8.87 \\
                                    & 4000 & 105249 & 512521 & 4.87 & 747617 & 7.1 & 1159641 & 11.02 & 937369 & 8.91 \\
                                    & 8000 & 211885 & 1010347 & 4.77 & 1526509 & 7.2 & 2561545 & 12.09 & 1911169 & 9.02 \\
                                    & 16000 & 426531 & 2028761 & 4.76 & 3093869 & 7.25 & 5578187 & 13.08 & 3912707 & 9.17 \\
                                    & 32000 & 862789 & 4065635 & 4.71 & 6376129 & 7.39 & 12145149 & 14.08 & 8118607 & 9.41 \\
\hline

\end{tabular}}

\centering
  \caption{Size comparison of incremental Delaunay complexes for various datasets (full version of Table~\ref{table:size}).}
  \label{table:size_full}
\end{table}

\begin{table}[h]
  \resizebox{\textwidth}{!}{
  \begin{tabular}{cc||ccc|ccc|ccc|ccc}
    \hline
    \multirow{3}{*}{Sample} & \multirow{3}{*}{\#Points} & \multicolumn{3}{c}{Random} & \multicolumn{3}{c}{Density} & \multicolumn{3}{c}{Height} & \multicolumn{3}{c}{Eccentricity} \\
    & & Time &  Time/simplex  & Memory & Time & Time/simplex & Memory & Time & Time/simplex & Memory & Time & Time/simplex & Memory \\
    & & (s) &   (ms) & (GB) & (s) & (ms) & (MB) & (s) & (ms) & (MB) & (s) & (ms) & (MB) \\
\hline
\multirow{6}{*}{$[0, 1]^2$}         &  1000 & 0.03 & 1.74 & 8.57 & 0.04 & 1.9 & 9.38 & 0.04 & 1.52 & 10.72 & 0.05 & 1.39 & 12.3 \\
                                    &  2000 & 0.08 & 2.02 & 13.18 & 0.08 & 1.55 & 14.95 & 0.1 & 1.64 & 18.3 & 0.12 & 1.66 & 20.91 \\
                                    &  4000 & 0.13 & 1.7 & 21.73 & 0.16 & 1.63 & 25.94 & 0.21 & 1.53 & 35.36 & 0.25 & 1.64 & 38.74 \\
                                    &  8000 & 0.28 & 1.79 & 39.18 & 0.36 & 1.76 & 48.2 & 0.47 & 1.55 & 73.05 & 0.54 & 1.65 & 76.29 \\
                                    &  16000 & 0.67 & 2.09 & 74.78 & 0.83 & 2.03 & 93.82 & 1.13 & 1.7 & 155.48 & 1.21 & 1.77 & 156.41 \\
                                    &  32000 & 1.53 & 2.39 & 147.96 & 1.79 & 2.13 & 187.54 & 2.55 & 1.78 & 322.46 & 2.64 & 1.88 & 322.64 \\

\hline
\multirow{6}{*}{$[0, 1]^3$}         &  1000 & 0.25 & 2.04 & 31.75 & 0.28 & 1.77 & 40.07 & 0.34 & 1.73 & 50.43 & 0.45 & 1.84 & 60.09 \\
                                    &  2000 & 0.52 & 2.05 & 62.64 & 0.73 & 2.12 & 81.86 & 0.89 & 1.95 & 110.41 & 1.11 & 2.02 & 128.71 \\
                                    &  4000 & 1.2 & 2.28 & 128.02 & 1.61 & 2.21 & 170.88 & 2.16 & 2.05 & 247.88 & 2.55 & 2.14 & 272.34 \\
                                    &  8000 & 2.63 & 2.46 & 246.07 & 3.49 & 2.3 & 342.96 & 5.1 & 2.14 & 554.64 & 5.64 & 2.2 & 586.6 \\
                                    &  16000 & 5.8 & 2.67 & 517.54 & 7.82 & 2.47 & 727.03 & 11.74 & 2.22 & 1166.0 & 12.49 & 2.31 & 1196.14 \\
                                    &  32000 & 12.36 & 2.84 & 1005.29 & 16.79 & 2.54 & 1505.86 & 27.07 & 2.32 & 2688.1 & 27.07 & 2.4 & 2576.14 \\

\hline
\multirow{6}{*}{$S^1$}              &  1000 & 0.03 & 1.58 & 8.44 & 0.05 & 2.01 & 9.4 & 0.04 & 1.5 & 9.72 & 0.05 & 1.35 & 13.12 \\
                                    &  2000 & 0.06 & 1.62 & 12.87 & 0.08 & 1.54 & 15.08 & 0.08 & 1.33 & 17.14 & 0.14 & 1.67 & 22.79 \\
                                    &  4000 & 0.13 & 1.69 & 21.88 & 0.16 & 1.58 & 26.48 & 0.17 & 1.36 & 32.86 & 0.3 & 1.66 & 44.19 \\
                                    &  8000 & 0.28 & 1.78 & 39.56 & 0.35 & 1.68 & 50.13 & 0.44 & 1.55 & 68.38 & 0.64 & 1.71 & 87.75 \\
                                    &  16000 & 0.64 & 2.02 & 74.72 & 0.81 & 1.87 & 98.46 & 1.03 & 1.66 & 146.57 & 1.44 & 1.82 & 178.85 \\
                                    &  32000 & 1.5 & 2.35 & 147.81 & 1.79 & 2.01 & 199.02 & 2.37 & 1.76 & 311.03 & 3.17 & 1.94 & 371.64 \\

\hline
\multirow{6}{*}{$S^2$}              &  1000 & 0.21 & 1.76 & 30.85 & 0.26 & 1.74 & 37.81 & 0.33 & 1.86 & 44.09 & 0.47 & 1.8 & 63.72 \\
                                    &  2000 & 0.48 & 1.93 & 59.57 & 0.67 & 2.07 & 77.06 & 0.79 & 1.89 & 101.46 & 1.25 & 2.09 & 139.36 \\
                                    &  4000 & 1.11 & 2.2 & 118.37 & 1.54 & 2.16 & 163.22 & 1.94 & 2.02 & 225.22 & 3.0 & 2.23 & 309.19 \\
                                    &  8000 & 2.5 & 2.4 & 244.05 & 3.57 & 2.33 & 354.23 & 4.57 & 2.11 & 496.82 & 6.95 & 2.33 & 681.24 \\
                                    &  16000 & 5.49 & 2.61 & 487.78 & 8.02 & 2.48 & 741.83 & 10.56 & 2.18 & 1093.43 & 15.6 & 2.44 & 1501.82 \\
                                    &  32000 & 12.0 & 2.81 & 964.62 & 17.56 & 2.57 & 1580.05 & 24.68 & 2.29 & 2469.64 & 34.62 & 2.52 & 3160.35 \\

\hline
\multirow{6}{*}{$S^1 \times S^1$}   &  1000 & 0.21 & 1.73 & 31.89 & 0.3 & 1.77 & 43.03 & 0.41 & 1.79 & 57.21 & 0.38 & 1.76 & 53.09 \\
                                    &  2000 & 0.49 & 1.92 & 62.76 & 0.75 & 2.06 & 85.69 & 1.05 & 2.03 & 125.26 & 0.95 & 2.04 & 110.82 \\
                                    &  4000 & 1.11 & 2.17 & 122.72 & 1.59 & 2.12 & 170.57 & 2.44 & 2.1 & 265.98 & 1.98 & 2.11 & 224.83 \\
                                    &  8000 & 2.36 & 2.34 & 236.03 & 3.49 & 2.29 & 354.77 & 5.54 & 2.16 & 619.1 & 4.17 & 2.18 & 437.08 \\
                                    &  16000 & 5.11 & 2.52 & 455.98 & 7.41 & 2.4 & 710.66 & 12.7 & 2.28 & 1271.56 & 8.94 & 2.28 & 910.05 \\
                                    &  32000 & 11.1 & 2.73 & 929.58 & 15.65 & 2.45 & 1461.12 & 28.5 & 2.35 & 2796.81 & 19.08 & 2.35 & 1847.01 \\

 \hline
\end{tabular}}
\centering
  \caption{Running time (in seconds), the average number of simplices processed per second, and the memory consumption (in MB) for computing the sublevel Delaunay bifiltration
(full version of Table~\ref{table:time_overview}).}
  \label{table:time_overview_full}
\end{table}

\subparagraph{Detailed runtime analysis.}
We measured the main sub-steps of the computation. There are three steps with significant running time,
and we display their contribution to the runtime in Table~\ref{table:time_analysis}:
the time to construct the incremental Delaunay complex, the time to compute all meb radii, and the time to generate
the chain complex out of the simplicial complex. The last step is not trivial because computing the boundary maps
requires a switch from
representing simplices by their boundary vertices (as done in the Simplex tree implementation of \textsc{Gudhi})
to representing simplices by their faces in codimension 1. We observe in Table~\ref{table:time_analysis}
that this step typically requires around half of the running time. We speculate that a different data structure
for simplicial complexes, more suitable for our application, could improve the running time further.

We further see that the computation of minimum enclosing ball radii takes generally more time than constructing
the complex itself. We point out here that the tools to compute mebs are not optimized for our application scenario,
where the number of points is only slightly larger than the dimension. There might also be room for improvement
here using a different computation method.

\begin{table}[!h]
  \resizebox{0.7\textwidth}{!}{
  \begin{tabular}{ccc||cccc}
    \hline
    Sample & Function & \#Points & Overall & Complex & meb & chain complex \\
\hline
    \multirow{6}{*}{$[0, 1]^2$} &\multirow{6}{*}{density}       & 1000 & 0.04 & 23.01\% & 29.67\% & 37.57\% \\
                                                       &        & 2000 & 0.08 & 15.59\% & 35.24\% & 43.44\% \\
                                                       &        & 4000 & 0.16 & 16.44\% & 33.82\% & 44.58\% \\
                                                       &        & 8000 & 0.36 & 16.12\% & 33.29\% & 45.28\% \\
                                                       &        & 16000 & 0.83 & 18.54\% & 31.85\% & 43.61\% \\
                                                       &        & 32000 & 1.79 & 18.16\% & 31.34\% & 45.11\% \\

    \hline
\multirow{6}{*}{$S^2$} &\multirow{6}{*}{eccentricity}           & 1000 & 0.47 & 12.81\% & 35.02\% & 50.99\% \\
                       &                                        & 2000 & 1.25 & 13.57\% & 33.94\% & 49.05\% \\
                       &                                        & 4000 & 3.0 & 11.90\% & 33.47\% & 51.07\% \\
                       &                                        & 8000 & 6.95 & 11.01\% & 32.66\% & 52.78\% \\
                       &                                        & 16000 & 15.6 & 11.03\% & 31.50\% & 53.81\% \\
                       &                                        & 32000 & 34.62 & 11.09\% & 31.14\% & 54.12\% \\

    \hline
\multirow{6}{*}{$S^1 \times S^1$} & \multirow{6}{*}{random}     & 1000 & 0.21 & 13.89\% & 34.83\% & 49.76\% \\
                                  &                             & 2000 & 0.49 & 13.19\% & 34.06\% & 50.01\% \\
                                  &                             & 4000 & 1.11 & 12.17\% & 33.90\% & 49.84\% \\
                                  &                             & 8000 & 2.36 & 12.07\% & 33.67\% & 50.50\% \\
                                  &                             & 16000 & 5.11 & 12.20\% & 32.96\% & 50.91\% \\
                                  &                             & 32000 & 11.1 & 12.09\% & 33.22\% & 50.72\% \\
\end{tabular}}
\centering
  \caption{The time breakdown (in seconds) for computing the sublevel Delaunay bifiltration.}
  \label{table:time_analysis}
\end{table}

\subparagraph{Higher dimensions.}
Our code can handle point clouds in all dimensions. We display the effect of higher dimensionality in Table~\ref{table:dimension}.
We observe that the size increases significantly when increasing the dimension which is expected since the worst-case bound
has an exponential dependence on $d$. However, we observe that the size of the Delaunay triangulation grows at a similar speed,
hence the overhead imposed by taking the function $\gamma$ into account is still relatively small.

\begin{table}[!h]
    \centering
  \resizebox{0.8\textwidth}{!}{
  \begin{tabular}{ccc || ccc | cc }
    \hline
    Sample & Function & Dim & Size & Time & Memory & Delaunay size & Ratio\\
                                                                          \hline
    \multirow{5}{*}{$[0, 1]^d$, n=500} & \multirow{5}{*}{density}       & 2 & 10917 & 0.02 & 6.58 & 2977 & 3.67 \\
                                       &                                & 3 & 70335 & 0.13 & 20.48 & 13505 & 5.21 \\
                                       &                                & 4 & 463649 & 1.11 & 109.95 & 77747 & 5.96 \\
                                       &                                & 5 & 3134645 & 10.36 & 1575.41 & 503747 & 6.22 \\
                                       &                                & 6 & 20632741 & 78.45 & 10246.52 & 3429063 & 6.02 \\
                                                                          \hline
    \multirow{3}{*}{$S^3$,n=500} & \multirow{3}{*}{height}              & 4 & 503865 & 1.14 & 121.46 & 74429 & 6.77 \\
                                 &                                      & 5 & 2417037 & 7.44 & 1102.71 & 319589 & 7.56 \\
                                 &                                      & 6 & 16089733 & 113.62 & 7379.98 & 1699797 & 9.47 \\
                                                                          \hline
    \multirow{5}{*}{$S^1\times S^1$, n=500} & \multirow{5}{*}{random}   & 3 & 61847 & 0.13 & 18.5 & 13471 & 4.59 \\
                                            &                           & 4 & 275805 & 0.65 & 67.02 & 60121 & 4.59 \\
                                            &                           & 5 & 1602905 & 5.53 & 773.18 & 301171 & 5.32 \\
                                            &                           & 6 & 7980373 & 31.01 & 3971.36 & 1499809 & 5.32 \\
    \hline
\end{tabular}}
\centering
  \caption{The size of the sublevel Delaunay bifiltration, the running time (in seconds) and memory consumption (in MB). Furthermore, the size of the Delaunay triangulation
and the ratio of the sizes are displayed.}
  \label{table:dimension}
\end{table}

\subparagraph{Compression.}
While not exploding in size for low dimension, the sublevel Delaunay bifiltration is still large in size
which poses a challenge to subsequent analysis algorithms on such a bifiltration.
Recall that we encode the bifiltration as a chain complex, where the chain group in dimension $k$
has the $k$-simplices of the incremental Delaunay complex as basis. The \emph{multi-chunk} algorithm\footnote{Available at \url{https://bitbucket.org/mkerber/multi_chunk/src/master/}}~\cite{fk-chunk,fkr-compression}
replaces a chain complex with a smaller chain complex (whose basis elements no longer correspond to simplices)
that is homotopy-equivalent. That is, as long as only topological properties of the bifiltration
(such as homology) are subsequently studied, we can post-process with the multi-chunk algorithm
to compress the resulting chain complex.

Table~\ref{table:multi-chunk} shows the results of this post-processing step. We can see the multi-chunk
significantly reduces the data size. Depending on the example, the compression
rate was between 0.0297 and 0.593.
The running time of multi-chunk varies a lot from instance to instance; this has to do with the special topological configuration
of the sublevel sets and we omit a detailed explanation. Finally, we can observe that the ratio between the post-processed
chain complex and the Delaunay triangulation appears to be non-increasing in many cases. This brings the question
whether it is possible to compute the compressed version of the sublevel Delaunay bifiltration directly,
without computing the uncompressed version first. Such an algorithm could potentially scale linearly in practice,
just like the computation of the Delaunay computation.

\begin{table}[h]\centering
  \resizebox{0.8\textwidth}{!}{
  \begin{tabular}{ccc||cc|ccc|cc}
    \hline
    \multirow{2}{*}{Sample} & \multirow{2}{*}{Function} & \multirow{2}{*}{\#Points} & \multicolumn{2}{c}{Complex} & \multicolumn{3}{c}{Multi-chunk} & \multirow{2}{*}{Delaunay size} & \multirow{2}{*}{Ratio} \\
    & & & Size & Time & Size & Compression & Time & & \\
 \hline
\multirow{6}{*}{$[0, 1]^2$} & \multirow{6}{*}{random}       &  1000 & 19223 & 0.0 & 10979 & 0.57 & 0.0 & 5971 & 1.84 \\
     &                                                      &  2000 & 39801 & 0.01 & 22385 & 0.56 & 0.01 & 11971 & 1.87 \\
     &                                                      &  4000 & 78859 & 0.02 & 45381 & 0.58 & 0.01 & 23969 & 1.89 \\
     &                                                      &  8000 & 158443 & 0.05 & 92339 & 0.58 & 0.03 & 47963 & 1.92 \\
     &                                                      &  16000 & 319753 & 0.11 & 187871 & 0.59 & 0.08 & 95951 & 1.96 \\
     &                                                      &  32000 & 639235 & 0.28 & 376459 & 0.59 & 0.14 & 191959 & 1.96 \\

    \hline
\multirow{6}{*}{$S^2$} & \multirow{6}{*}{density}           &  1000 & 150747 & 0.04 & 14223 & 0.09 & 0.02 & 26439 & 0.54 \\
     &                                                      &  2000 & 325791 & 0.1 & 29503 & 0.09 & 0.06 & 54327 & 0.54 \\
     &                                                      &  4000 & 709303 & 0.18 & 62039 & 0.09 & 0.24 & 109987 & 0.56 \\
     &                                                      &  8000 & 1534939 & 0.46 & 127949 & 0.08 & 0.86 & 222339 & 0.58 \\
     &                                                      &  16000 & 3241029 & 0.92 & 268187 & 0.08 & 3.17 & 450133 & 0.6 \\
     &                                                      &  32000 & 6824453 & 2.06 & 557371 & 0.08 & 10.62 & 909567 & 0.61 \\

    \hline
\multirow{6}{*}{$S^1\times S^1$} & \multirow{6}{*}{height}  &  1000 & 230075 & 0.05 & 10793 & 0.05 & 0.03 & 27003 & 0.4 \\
     &                                                      &  2000 & 515437 & 0.12 & 21305 & 0.04 & 0.05 & 52717 & 0.4 \\
     &                                                      &  4000 & 1159641 & 0.29 & 40501 & 0.03 & 0.14 & 105249 & 0.38 \\
     &                                                      &  8000 & 2561545 & 0.65 & 81615 & 0.03 & 0.34 & 211885 & 0.39 \\
     &                                                      &  16000 & 5578187 & 1.41 & 167225 & 0.03 & 0.78 & 426531 & 0.39 \\
     &                                                      &  32000 & 12145149 & 3.15 & 357835 & 0.03 & 1.66 & 862789 & 0.41 \\
    \hline
\end{tabular}}
\centering
  \caption{The size of the sublevel Delaunay bifiltration, the running time (in seconds) and memory consumption (in MB). Furthermore, the size of the Delaunay triangulation
and the ratio of the sizes are displayed.}
  \label{table:multi-chunk}
\end{table}

\end{appendix}

\bibliographystyle{plainurl_fulljournal}
\bibliography{refs}

\begin{thebibliography}{10}

\bibitem{adcock2014classification}
Aaron Adcock, Daniel Rubin, and Gunnar Carlsson.
\newblock Classification of hepatic lesions using the matching metric.
\newblock {\em Computer vision and image understanding}, 121:36--42, 2014.
\newblock \href {https://doi.org/10.1016/j.cviu.2013.10.014}
  {\path{doi:10.1016/j.cviu.2013.10.014}}.

\bibitem{alonso2023filtration}
{\'A}ngel~Javier Alonso, Michael Kerber, and Siddharth Pritam.
\newblock Filtration-domination in bifiltered graphs.
\newblock In {\em 2023 Proceedings of the Symposium on Algorithm Engineering
  and Experiments (ALENEX)}, pages 27--38. SIAM, 2023.

\bibitem{asashiba2019approximation}
Hideto Asashiba, Emerson~G. Escolar, Ken Nakashima, and Michio Yoshiwaki.
\newblock On approximation of 2d persistence modules by interval-decomposables.
\newblock {\em Journal of Computational Algebra}, 6-7:100007, 2023.
\newblock \href {https://doi.org/10.1016/j.jaca.2023.100007}
  {\path{doi:10.1016/j.jaca.2023.100007}}.

\bibitem{bauerRipserEfficientComputation2021}
Ulrich Bauer.
\newblock Ripser: efficient computation of {V}ietoris-{R}ips persistence
  barcodes.
\newblock {\em Journal of Applied and Computational Topology}, 5(3):391--423,
  2021.
\newblock \href {https://doi.org/10.1007/s41468-021-00071-5}
  {\path{doi:10.1007/s41468-021-00071-5}}.

\bibitem{bauerMorseTheoryCech2016}
Ulrich Bauer and Herbert Edelsbrunner.
\newblock The {M}orse theory of \v{C}ech and {D}elaunay complexes.
\newblock {\em Transactions of the American Mathematical Society},
  369(5):3741--3762, 2017.
\newblock \href {https://doi.org/10.1090/tran/6991}
  {\path{doi:10.1090/tran/6991}}.

\bibitem{bauer2023unified}
Ulrich Bauer, Michael Kerber, Fabian Roll, and Alexander Rolle.
\newblock A unified view on the functorial nerve theorem and its variations.
\newblock {\em Expositiones Mathematicae}, 41(4):125503, 2023.
\newblock \href {https://doi.org/10.1016/j.exmath.2023.04.005}
  {\path{doi:10.1016/j.exmath.2023.04.005}}.

\bibitem{bauer2023efficient}
Ulrich Bauer, Fabian Lenzen, and Michael Lesnick.
\newblock {Efficient Two-Parameter Persistence Computation via Cohomology}.
\newblock In {\em 39th International Symposium on Computational Geometry (SoCG
  2023)}.
\newblock \href {https://doi.org/10.4230/LIPIcs.SoCG.2023.15}
  {\path{doi:10.4230/LIPIcs.SoCG.2023.15}}.

\bibitem{benjamin2022multiscale}
Katherine Benjamin, Aneesha Bhandari, Zhouchun Shang, Yanan Xing, Yanru An,
  Nannan Zhang, Yong Hou, Ulrike Tillmann, Katherine~R. Bull, and Heather~A.
  Harrington.
\newblock Multiscale topology classifies and quantifies cell types in
  subcellular spatial transcriptomics, 2022.
\newblock \href {https://arxiv.org/abs/2212.06505} {\path{arXiv:2212.06505}}.

\bibitem{biswas2022size}
Ranita Biswas, Sebastiano~Cultrera di~Montesano, Ondřej Draganov, Herbert
  Edelsbrunner, and Morteza Saghafian.
\newblock On the size of chromatic {D}elaunay mosaics, 2022.
\newblock \href {https://arxiv.org/abs/2212.03121} {\path{arXiv:2212.03121}}.

\bibitem{bjerkevik2021ell}
Håvard~Bakke Bjerkevik and Michael Lesnick.
\newblock $\ell^p$-distances on multiparameter persistence modules, 2021.
\newblock \href {https://arxiv.org/abs/2106.13589} {\path{arXiv:2106.13589}}.

\bibitem{blanchette2021homological}
Benjamin Blanchette, Thomas Brüstle, and Eric~J. Hanson.
\newblock Homological approximations in persistence theory.
\newblock {\em Canadian Journal of Mathematics}, pages 1--38, 2022.
\newblock \href {https://doi.org/10.4153/s0008414x22000657}
  {\path{doi:10.4153/s0008414x22000657}}.

\bibitem{blumberg2022stability}
Andrew~J. Blumberg and Michael Lesnick.
\newblock Stability of 2-parameter persistent homology.
\newblock {\em Foundations of Computational Mathematics}.
\newblock \href {https://doi.org/10.1007/s10208-022-09576-6}
  {\path{doi:10.1007/s10208-022-09576-6}}.

\bibitem{blumberg2017universality}
Andrew~J. Blumberg and Michael Lesnick.
\newblock Universality of the homotopy interleaving distance.
\newblock {\em Transactions of the American Mathematical Society, {$\mathrm{in\
  press}$}}.
\newblock \href {https://arxiv.org/abs/1705.01690} {\path{arXiv:1705.01690}}.

\bibitem{boissonnat2018geometric}
Jean-Daniel Boissonnat, Fr\'{e}d\'{e}ric Chazal, and Mariette Yvinec.
\newblock {\em Geometric and topological inference}.
\newblock Cambridge University Press, 2018.
\newblock \href {https://doi.org/10.1017/9781108297806}
  {\path{doi:10.1017/9781108297806}}.

\bibitem{boissonat2009incremental}
Jean-Daniel Boissonnat, Olivier Devillers, and Samuel Hornus.
\newblock Incremental construction of the {D}elaunay triangulation and the
  {D}elaunay graph in medium dimension.
\newblock In {\em Proceedings of the Twenty-Fifth Annual Symposium on
  Computational Geometry}, SCG '09, page 208–216, 2009.
\newblock \href {https://doi.org/10.1145/1542362.1542403}
  {\path{doi:10.1145/1542362.1542403}}.

\bibitem{boissonnat2000triangulations}
Jean-Daniel Boissonnat, Olivier Devillers, Sylvain Pion, Monique Teillaud, and
  Mariette Yvinec.
\newblock Triangulations in {CGAL}.
\newblock {\em Computational Geometry. Theory and Applications}, 22(1-3):5--19,
  2002.
\newblock \href {https://doi.org/10.1016/S0925-7721(01)00054-2}
  {\path{doi:10.1016/S0925-7721(01)00054-2}}.

\bibitem{simplex_tree_paper}
Jean{-}Daniel Boissonnat and Cl{\'{e}}ment Maria.
\newblock The simplex tree: An efficient data structure for general simplicial
  complexes.
\newblock {\em Algorithmica}, 70(3):406--427, 2014.
\newblock \href {https://doi.org/10.1007/s00453-014-9887-3}
  {\path{doi:10.1007/s00453-014-9887-3}}.

\bibitem{boissonnat2020edge}
Jean-Daniel Boissonnat and Siddharth Pritam.
\newblock Edge collapse and persistence of flag complexes.
\newblock In {\em 36th International Symposium on Computational Geometry (SoCG
  2020)}.
\newblock \href {https://doi.org/10.4230/LIPIcs.SoCG.2020.19}
  {\path{doi:10.4230/LIPIcs.SoCG.2020.19}}.

\bibitem{bt-randomized}
Jean-Daniel Boissonnat and Monique Teillaud.
\newblock On the randomized construction of the {D}elaunay tree.
\newblock {\em Theoretical Computer Science}, 112(2):339--354, 1993.
\newblock \href {https://doi.org/10.1016/0304-3975(93)90024-N}
  {\path{doi:10.1016/0304-3975(93)90024-N}}.

\bibitem{botnan2022introduction}
Magnus~Bakke Botnan and Michael Lesnick.
\newblock An introduction to multiparameter persistence.
\newblock {\em Proceedings of the 2020 International Conference on
  Representations of Algebras, $\mathrm{in\ press}$}.
\newblock \href {https://arxiv.org/abs/2203.14289} {\path{arXiv:2203.14289}}.

\bibitem{botnan2021signed}
Magnus~Bakke Botnan, Steffen Oppermann, and Steve Oudot.
\newblock {Signed Barcodes for Multi-Parameter Persistence via Rank
  Decompositions}.
\newblock In {\em 38th International Symposium on Computational Geometry (SoCG
  2022)}.
\newblock \href {https://doi.org/10.4230/LIPIcs.SoCG.2022.19}
  {\path{doi:10.4230/LIPIcs.SoCG.2022.19}}.

\bibitem{botnan2022bottleneck}
Magnus~Bakke Botnan, Steffen Oppermann, Steve Oudot, and Luis Scoccola.
\newblock On the bottleneck stability of rank decompositions of multi-parameter
  persistence modules, 2022.
\newblock \href {https://arxiv.org/abs/2208.00300} {\path{arXiv:2208.00300}}.

\bibitem{bowyerComputingDirichletTessellations1981}
Adrian Bowyer.
\newblock Computing {D}irichlet tessellations.
\newblock {\em The Computer Journal}, 24(2):162--166, 1981.
\newblock \href {https://doi.org/10.1093/comjnl/24.2.162}
  {\path{doi:10.1093/comjnl/24.2.162}}.

\bibitem{cai2021elder}
Chen Cai, Woojin Kim, Facundo M\'{e}moli, and Yusu Wang.
\newblock Elder-rule-staircodes for augmented metric spaces.
\newblock {\em SIAM Journal on Applied Algebra and Geometry}, 5(3):417--454,
  2021.
\newblock \href {https://doi.org/10.1137/20M1353605}
  {\path{doi:10.1137/20M1353605}}.

\bibitem{cang2018representability}
Zixuan Cang, Lin Mu, and Guo-Wei Wei.
\newblock Representability of algebraic topology for biomolecules in machine
  learning based scoring and virtual screening.
\newblock {\em PLoS computational biology}, 14(1):e1005929, 2018.
\newblock \href {https://doi.org/10.1371/journal.pcbi.1005929}
  {\path{doi:10.1371/journal.pcbi.1005929}}.

\bibitem{cang2020persistent}
Zixuan Cang and Guo-Wei Wei.
\newblock Persistent cohomology for data with multicomponent heterogeneous
  information.
\newblock {\em SIAM Journal on Mathematics of Data Science}, 2(2):396--418,
  2020.
\newblock \href {https://doi.org/10.1137/19M1272226}
  {\path{doi:10.1137/19M1272226}}.

\bibitem{carlsson2008local}
Gunnar Carlsson, Tigran Ishkhanov, Vin De~Silva, and Afra Zomorodian.
\newblock {On the local behavior of spaces of natural images}.
\newblock {\em International Journal of Computer Vision}, 76(1):1--12, 2008.
\newblock \href {https://doi.org/10.1007/s11263-007-0056-x}
  {\path{doi:10.1007/s11263-007-0056-x}}.

\bibitem{carlssonTheoryMultidimensionalPersistence2009}
Gunnar Carlsson and Afra Zomorodian.
\newblock The theory of multidimensional persistence.
\newblock {\em Discrete \& Computational Geometry}, 42(1):71--93, 2009.
\newblock \href {https://doi.org/10.1007/s00454-009-9176-0}
  {\path{doi:10.1007/s00454-009-9176-0}}.

\bibitem{carriere2020multiparameter}
Mathieu Carri{\`e}re and Andrew Blumberg.
\newblock Multiparameter persistence images for topological machine learning.
\newblock In {\em Conference on Neural Information Processing Systems
  (NIPS'20)}, page 22432–22444.

\bibitem{chariDiscreteMorseFunctions2000}
Manoj~K. Chari.
\newblock On discrete {M}orse functions and combinatorial decompositions.
\newblock {\em Discrete Mathematics}, 217(1-3):101--113, 2000.
\newblock \href {https://doi.org/10.1016/S0012-365X(99)00258-7}
  {\path{doi:10.1016/S0012-365X(99)00258-7}}.

\bibitem{corbet2019kernel}
Ren{\'e} Corbet, Ulderico Fugacci, Michael Kerber, Claudia Landi, and Bei Wang.
\newblock A kernel for multi-parameter persistent homology.
\newblock {\em Computers \& Graphics: X}, 2:100005, 2019.
\newblock \href {https://doi.org/10.1016/j.cagx.2019.100005}
  {\path{doi:10.1016/j.cagx.2019.100005}}.

\bibitem{corbet2023computing}
Ren\'{e} Corbet, Michael Kerber, Michael Lesnick, and Georg Osang.
\newblock Computing the multicover bifiltration.
\newblock {\em Discrete \& Computational Geometry}, 70(2):376--405, 2023.
\newblock \href {https://doi.org/10.1007/s00454-022-00476-8}
  {\path{doi:10.1007/s00454-022-00476-8}}.

\bibitem{dutch_book}
Mark de~Berg, Otfried Cheong, Marc~J. van Kreveld, and Mark~H. Overmars.
\newblock {\em Computational geometry: algorithms and applications, 3rd
  Edition}.
\newblock Springer, 2008.
\newblock \href {https://doi.org/10.1007/978-3-540-77974-2}
  {\path{doi:10.1007/978-3-540-77974-2}}.

\bibitem{devillers-delaunay}
Olivier Devillers.
\newblock The {Delaunay} hierarchy.
\newblock {\em International Journal of Foundations of Computer Science}, pages
  163--180, 2002.
\newblock \href {https://doi.org/10.1142/S0129054102001035}
  {\path{doi:10.1142/S0129054102001035}}.

\bibitem{cgal_dt_dd}
Olivier Devillers, Samuel Hornus, and Cl{\'e}ment Jamin.
\newblock {dD} triangulations.
\newblock In {\em {CGAL} User and Reference Manual}. {CGAL Editorial Board},
  {5.5.2} edition, 2023.
\newblock URL:
  \url{https://doc.cgal.org/5.5.2/Manual/packages.html#PkgTriangulations}.

\bibitem{dpt-walking}
Olivier Devillers, Sylvain Pion, and Monique Teillaud.
\newblock Walking in a triangulation.
\newblock {\em International Journal of Foundations of Computer Science},
  13(02):181--199, 2002.
\newblock \href {https://doi.org/10.1142/S0129054102001047}
  {\path{doi:10.1142/S0129054102001047}}.

\bibitem{dey_et_al:LIPIcs.SoCG.2022.34}
Tamal~K. Dey, Woojin Kim, and Facundo M\'{e}moli.
\newblock {Computing Generalized Rank Invariant for 2-Parameter Persistence
  Modules via Zigzag Persistence and Its Applications}.
\newblock In {\em 38th International Symposium on Computational Geometry (SoCG
  2022)}.
\newblock \href {https://doi.org/10.4230/LIPIcs.SoCG.2022.34}
  {\path{doi:10.4230/LIPIcs.SoCG.2022.34}}.

\bibitem{dimontesano2022persistent}
Sebastiano~Cultrera di~Montesano, Ondřej Draganov, Herbert Edelsbrunner, and
  Morteza Saghafian.
\newblock Persistent homology of chromatic alpha complexes, 2022.
\newblock \href {https://arxiv.org/abs/2212.03128} {\path{arXiv:2212.03128}}.

\bibitem{edelsbrunnerComputationalTopologyIntroduction2010}
Herbert Edelsbrunner and John~L. Harer.
\newblock {\em Computational topology: an introduction}.
\newblock American Mathematical Society, Providence, RI, 2010.
\newblock \href {https://doi.org/10.1090/mbk/069} {\path{doi:10.1090/mbk/069}}.

\bibitem{edelsbrunner2021multi}
Herbert Edelsbrunner and Georg Osang.
\newblock The multi-cover persistence of {E}uclidean balls.
\newblock {\em Discrete \& Computational Geometry}, 65(4):1296--1313, 2021.
\newblock \href {https://doi.org/10.1007/s00454-021-00281-9}
  {\path{doi:10.1007/s00454-021-00281-9}}.

\bibitem{edelsbrunner2023simple}
Herbert Edelsbrunner and Georg Osang.
\newblock A simple algorithm for higher-order {D}elaunay mosaics and alpha
  shapes.
\newblock {\em Algorithmica}, 85(1):277--295, 2023.
\newblock \href {https://doi.org/10.1007/s00453-022-01027-6}
  {\path{doi:10.1007/s00453-022-01027-6}}.

\bibitem{edelsbrunner1996incremental}
Herbert Edelsbrunner and Nimish~R. Shah.
\newblock Incremental topological flipping works for regular triangulations.
\newblock {\em Algorithmica}, 15(3):223--241, 1996.
\newblock \href {https://doi.org/10.1007/s004539900013}
  {\path{doi:10.1007/s004539900013}}.

\bibitem{cgal_meb}
Kaspar Fischer, Bernd G{\"a}rtner, Thomas Herrmann, Michael Hoffmann, and Sven
  Sch{\"o}nherr.
\newblock Bounding volumes.
\newblock In {\em {CGAL} User and Reference Manual}. {CGAL Editorial Board},
  {5.5.2} edition, 2023.
\newblock URL:
  \url{https://doc.cgal.org/5.5.2/Manual/packages.html#PkgBoundingVolumes}.

\bibitem{formanMorseTheoryCell1998}
Robin Forman.
\newblock Morse theory for cell complexes.
\newblock {\em Advances in Mathematics}, 134(1):90--145, 1998.
\newblock \href {https://doi.org/10.1006/aima.1997.1650}
  {\path{doi:10.1006/aima.1997.1650}}.

\bibitem{fk-chunk}
Ulderico Fugacci and Michael Kerber.
\newblock Chunk reduction for multi-parameter persistent homology.
\newblock In {\em 35th International Symposium on Computational Geometry (SoCG
  2019)}.
\newblock \href {https://doi.org/10.4230/LIPIcs.SoCG.2019.37}
  {\path{doi:10.4230/LIPIcs.SoCG.2019.37}}.

\bibitem{fkr-compression}
Ulderico Fugacci, Michael Kerber, and Alexander Rolle.
\newblock Compression for 2-parameter persistent homology.
\newblock {\em Computational Geometry}, 109:101940, 2023.
\newblock \href {https://doi.org/10.1016/j.comgeo.2022.101940}
  {\path{doi:10.1016/j.comgeo.2022.101940}}.

\bibitem{gaertner-esa}
Bernd G{\"{a}}rtner.
\newblock Fast and robust smallest enclosing balls.
\newblock In {\em Algorithms - {ESA} '99, 7th Annual European Symposium}.
\newblock \href {https://doi.org/10.1007/3-540-48481-7\_29}
  {\path{doi:10.1007/3-540-48481-7\_29}}.

\bibitem{glisse2022swap}
Marc Glisse and Siddharth Pritam.
\newblock {Swap, Shift and Trim to Edge Collapse a Filtration}.
\newblock In {\em 38th International Symposium on Computational Geometry (SoCG
  2022)}.
\newblock \href {https://doi.org/10.4230/LIPIcs.SoCG.2022.44}
  {\path{doi:10.4230/LIPIcs.SoCG.2022.44}}.

\bibitem{cgal_dt_3d}
Cl{\'e}ment Jamin, Sylvain Pion, and Monique Teillaud.
\newblock {3D} triangulations.
\newblock In {\em {CGAL} User and Reference Manual}. {CGAL Editorial Board},
  {5.5.2} edition, 2023.
\newblock URL:
  \url{https://doc.cgal.org/5.5.2/Manual/packages.html#PkgTriangulation3}.

\bibitem{jung1901ueber}
Heinrich Jung.
\newblock Ueber die kleinste {K}ugel, die eine r\"{a}umliche {F}igur
  einschliesst.
\newblock {\em Journal f\"{u}r die Reine und Angewandte Mathematik},
  123:241--257, 1901.
\newblock \href {https://doi.org/10.1515/crll.1901.123.241}
  {\path{doi:10.1515/crll.1901.123.241}}.

\bibitem{keller2018persistent}
Bryn Keller, Michael Lesnick, and Theodore~L Willke.
\newblock Persistent homology for virtual screening.
\newblock 2018.
\newblock \href {https://doi.org/10.26434/chemrxiv.6969260.v3}
  {\path{doi:10.26434/chemrxiv.6969260.v3}}.

\bibitem{scc2020}
Michael Kerber and Michael Lesnick.
\newblock scc2020: A file format for sparse chain complexes in {TDA}, 2021.
\newblock URL: \url{https://bitbucket.org/mkerber/chain_complex_format/}.

\bibitem{kerber2020efficient}
Michael Kerber and Arnur Nigmetov.
\newblock {Efficient Approximation of the Matching Distance for 2-Parameter
  Persistence}.
\newblock In {\em 36th International Symposium on Computational Geometry (SoCG
  2020)}.
\newblock \href {https://doi.org/10.4230/LIPIcs.SoCG.2020.53}
  {\path{doi:10.4230/LIPIcs.SoCG.2020.53}}.

\bibitem{kr-fast}
Michael Kerber and Alexander Rolle.
\newblock Fast minimal presentations of bi-graded persistence modules.
\newblock In {\em Proceedings of the Symposium on Algorithm Engineering and
  Experiments, {ALENEX} 2021}.
\newblock \href {https://doi.org/10.1137/1.9781611976472.16}
  {\path{doi:10.1137/1.9781611976472.16}}.

\bibitem{kim2021generalized}
Woojin Kim and Facundo M\'{e}moli.
\newblock Generalized persistence diagrams for persistence modules over posets.
\newblock {\em Journal of Applied and Computational Topology}, 5(4):533--581,
  2021.
\newblock \href {https://doi.org/10.1007/s41468-021-00075-1}
  {\path{doi:10.1007/s41468-021-00075-1}}.

\bibitem{lanari2023rectification}
Edoardo Lanari and Luis Scoccola.
\newblock Rectification of interleavings and a persistent {W}hitehead theorem.
\newblock {\em Algebraic \& Geometric Topology}, 23(2):803--832, 2023.
\newblock \href {https://doi.org/10.2140/agt.2023.23.803}
  {\path{doi:10.2140/agt.2023.23.803}}.

\bibitem{lawsonPropertiesNdimensionalTriangulations1986}
Charles~L. Lawson.
\newblock Properties of {$n$}-dimensional triangulations.
\newblock {\em Computer Aided Geometric Design}, 3(4):231--246 (1987), 1986.
\newblock \href {https://doi.org/10.1016/0167-8396(86)90001-4}
  {\path{doi:10.1016/0167-8396(86)90001-4}}.

\bibitem{lesnick2015theory}
Michael Lesnick.
\newblock The theory of the interleaving distance on multidimensional
  persistence modules.
\newblock {\em Foundations of Computational Mathematics}, 15(3):613--650, 2015.
\newblock \href {https://doi.org/10.1007/s10208-015-9255-y}
  {\path{doi:10.1007/s10208-015-9255-y}}.

\bibitem{lesnick2015interactive}
Michael Lesnick and Matthew Wright.
\newblock Interactive visualization of 2-{D} persistence modules.
\newblock 2015.
\newblock \href {https://arxiv.org/abs/1512.00180} {\path{arXiv:1512.00180}}.

\bibitem{lw-computing}
Michael Lesnick and Matthew Wright.
\newblock Computing minimal presentations and bigraded {B}etti numbers of
  2-parameter persistent homology.
\newblock {\em SIAM Journal on Applied Algebra and Geometry}, 6(2):267--298,
  2022.
\newblock \href {https://doi.org/10.1137/20M1388425}
  {\path{doi:10.1137/20M1388425}}.

\bibitem{loiseaux2023stable}
David Loiseaux, Luis Scoccola, Mathieu Carrière, Magnus~Bakke Botnan, and
  Steve Oudot.
\newblock Stable vectorization of multiparameter persistent homology using
  signed barcodes as measures.
\newblock 2023.
\newblock \href {https://arxiv.org/abs/2306.03801} {\path{arXiv:2306.03801}}.

\bibitem{gudhi:FilteredComplexes}
Cl\'ement Maria.
\newblock Filtered complexes.
\newblock In {\em {GUDHI} User and Reference Manual}. {GUDHI Editorial Board},
  2015.
\newblock URL:
  \url{http://gudhi.gforge.inria.fr/doc/latest/group__simplex__tree.html}.

\bibitem{mccleary2022edit}
Alexander McCleary and Amit Patel.
\newblock Edit distance and persistence diagrams over lattices.
\newblock {\em SIAM Journal on Applied Algebra and Geometry}, 6(2):134--155,
  2022.
\newblock \href {https://doi.org/10.1137/20M1373700}
  {\path{doi:10.1137/20M1373700}}.

\bibitem{morozov2021output}
Dmitriy Morozov and Amit Patel.
\newblock Output-sensitive computation of generalized persistence diagrams for
  2-filtrations.
\newblock 2021.
\newblock \href {https://arxiv.org/abs/2112.03980} {\path{arXiv:2112.03980}}.

\bibitem{reaniCoupledAlphaComplex2021}
Yohai Reani and Omer Bobrowski.
\newblock A coupled alpha complex.
\newblock 2021.
\newblock \href {https://arxiv.org/abs/2105.08113} {\path{arXiv:2105.08113}}.

\bibitem{rolle-degree}
Alexander Rolle.
\newblock The degree-{R}ips complexes of an annulus with outliers.
\newblock In {\em 38th International Symposium on Computational Geometry, SoCG
  2022}.
\newblock \href {https://doi.org/10.4230/LIPIcs.SoCG.2022.58}
  {\path{doi:10.4230/LIPIcs.SoCG.2022.58}}.

\bibitem{gudhi:CechComplex}
Vincent Rouvreau and Hind Montassif.
\newblock {\v{C}}ech complex.
\newblock In {\em GUDHI User and Reference Manual}. GUDHI Editorial Board,
  3.8.0 edition, 2023.
\newblock URL:
  \url{https://gudhi.inria.fr/doc/3.8.0/group__cech__complex.html}.

\bibitem{scoccola2023persistable}
Luis Scoccola and Alexander Rolle.
\newblock Persistable: persistent and stable clustering.
\newblock {\em Journal of Open Source Software}, 8(83):5022, 2023.
\newblock \href {https://doi.org/10.21105/joss.05022}
  {\path{doi:10.21105/joss.05022}}.

\bibitem{sheehy2012multicover}
Donald~R. Sheehy.
\newblock A multicover nerve for geometric inference.
\newblock In {\em CCCG: Canadian Conference in Computational Geometry}, 2012.
\newblock URL: \url{http://2012.cccg.ca/e-proceedings.pdf}.

\bibitem{vipond2018multiparameter}
Oliver Vipond.
\newblock Multiparameter persistence landscapes.
\newblock {\em Journal of Machine Learning Research}, 21(61):1--38, 2020.
\newblock URL: \url{http://jmlr.org/papers/v21/19-054.html}.

\bibitem{vipond2021multiparameter}
Oliver Vipond, Joshua~A. Bull, Philip~S. Macklin, Ulrike Tillmann,
  Christopher~W. Pugh, Helen~M. Byrne, and Heather~A. Harrington.
\newblock Multiparameter persistent homology landscapes identify immune cell
  spatial patterns in tumors.
\newblock {\em Proceedings of the National Academy of Sciences},
  118(41):e2102166118, 2021.
\newblock \href {https://doi.org/10.1073/pnas.2102166118}
  {\path{doi:10.1073/pnas.2102166118}}.

\bibitem{watsonComputingNdimensionalDelaunay1981}
David~F. Watson.
\newblock Computing the {$n$}-dimensional {D}elaunay tessellation with
  application to {V}orono\u{\i} polytopes.
\newblock {\em The Computer Journal}, 24(2):167--172, 1981.
\newblock \href {https://doi.org/10.1093/comjnl/24.2.167}
  {\path{doi:10.1093/comjnl/24.2.167}}.

\bibitem{whiteheadCollapses}
J.~H.~C. Whitehead.
\newblock Simplicial {S}paces, {N}uclei and m-{G}roups.
\newblock {\em Proceedings of the London Mathematical Society. Second Series},
  45(4):243--327, 1939.
\newblock \href {https://doi.org/10.1112/plms/s2-45.1.243}
  {\path{doi:10.1112/plms/s2-45.1.243}}.

\bibitem{xin2023gril}
Cheng Xin, Soham Mukherjee, Shreyas~N. Samaga, and Tamal~K. Dey.
\newblock Gril: A $2$-parameter persistence based vectorization for machine
  learning.
\newblock 2023.
\newblock \href {https://arxiv.org/abs/2304.04970} {\path{arXiv:2304.04970}}.

\bibitem{cgal_dt_2d}
Mariette Yvinec.
\newblock {2D} triangulations.
\newblock In {\em {CGAL} User and Reference Manual}. {CGAL Editorial Board},
  {5.5.2} edition, 2023.
\newblock URL:
  \url{https://doc.cgal.org/5.5.2/Manual/packages.html#PkgTriangulation2}.

\end{thebibliography}

\end{document}